\documentclass[letterpaper,journal,12pt,draftclsnofoot,onecolumn]{IEEEtran}
\usepackage{amsmath,amsfonts}
\usepackage{algorithmic}
\usepackage{algorithm}
\usepackage{array}
\ifCLASSOPTIONcompsoc
\usepackage[caption=false,font=normalsize,labelfont=sf,textfont=sf]{subfig}
\else
\usepackage[caption=false,font=footnotesize]{subfig}
\fi
\usepackage{textcomp}
\usepackage{stfloats}
\usepackage{url}
\usepackage{verbatim}
\usepackage{graphicx}
\usepackage{cite}
\usepackage{commonmath}
\usepackage{globalglossary}
\usepackage{xcolor}
\usepackage{multirow}
\usepackage{amsthm}
\usepackage{caption}
\usepackage{cleveref}

\crefname{equation}{}{}

\theoremstyle{remark}
\newtheorem{lemma}{Lemma}
\newtheorem{proposition}{Proposition}
\newtheorem{remark}{Remark}
\newtheorem{corollary}{Corollary}
\glsaddkey*%
  {ConjTrans}%
  {\ensuremath{\glsentrytext{\glslabel}^H}}%
  {\glsentryConjTrans}%
  {\GlsentryConjTrans}%
  {\glsConjTrans}%
  {\GlsConjTrans}%
  {\GLSConjTrans}
\glsaddkey*%
  {Squared}%
  {\ensuremath{\glsentrytext{\glslabel}^2}}%
  {\glsentrySquared}%
  {\GlsentrySquared}%
  {\glsSquared}%
  {\GlsSquared}%
  {\GLSSquared}

\glsxtrnewsymbol[
  description={Transmit symbol vector}
]{TransmitSymbolVector}{\ensuremath{\vt{x}}}
\glsxtrnewsymbol[
  description={Number of RIS unit cells}
  ]{NumRis}{\ensuremath{Q}}
\glsxtrnewsymbol[
  description={Number of RIS unit cells}
  ]{NumRis2}{\ensuremath{Q^2}}
\glsxtrnewsymbol[
  description={Number of RIS unit cells}
  ]{NumRisX}{\ensuremath{Q_x}}
\glsxtrnewsymbol[
  description={Number of RIS unit cells}
  ]{NumRisY}{\ensuremath{Q_y}}
\glsxtrnewsymbol[
  description={Number of symbols at level}
]{NumSymTxLevel}{\ensuremath{N_{s,l}}}
\glsxtrnewsymbol[
  description={Number of symbols at next level}
]{NumSymTxLevelNext}{\ensuremath{N_{s,l+1}}}
\glsxtrnewsymbol[
  description={Number of symbols at level 1}
]{NumSymTxLevelZero}{\ensuremath{N_{s,1}}}
\glsxtrnewsymbol[
  description={Number of symbols at level L}
]{NumSymTxLevelMax}{\ensuremath{N_{s,L}}}
\glsxtrnewsymbol[
  description={Number of AP antennas}
  ]{NumAp}{\ensuremath{N_\mathrm{BS}}}
\glsxtrnewsymbol[
  description={Number of training symbols}
  ]{NumTraining}{\ensuremath{N_t}}
\glsxtrnewsymbol[
  description={Number of training symbols}
  ]{NumTrainingMax}{\ensuremath{N_s^\mathrm{max}}}
\glsxtrnewsymbol[
  description={Phase applied at a unit cell}
  ]{PhaseUc}{\ensuremath{\psi}}
\glsxtrnewsymbol[
  description={RIS gain including unit cell gain}
]{GainRis}{\ensuremath{g}}
\glsxtrnewsymbol[
  description={Gain of unit cell coefficient}
  ]{GainUc}{\ensuremath{g_\mathrm{UC}}}
\glsxtrnewsymbol[
  description={End-to-end channel gain}
  ]{GainEndToEnd}{\ensuremath{g^\prime}}
\glsxtrnewsymbol[
  description={End-to-end channel gain factor}
  ]{GainEndToEndFactor}{\ensuremath{\tilde{g}}}
\glsxtrnewsymbol[
  description={Effective gain factor of a unit cell}
  ]{GainUcEffective}{\ensuremath{\bar{g}}}
\glsxtrnewsymbol[
  description={RIS response function}
]{RisResponse}{\ensuremath{\tilde{g}}}
\glsxtrnewsymbol[
  description={Set of codewords at level l}
]{SetCodewords}{\ensuremath{\mathcal{M}_l}}
\glsxtrnewsymbol[
  description={Set of codewords at level 1}
]{SetCodewordsInitial}{\ensuremath{\mathcal{M}_1}}
\glsxtrnewsymbol[
  description={Set of codewords at level L}
]{SetCodewordsFinal}{\ensuremath{\mathcal{M}_L}}
\glsxtrnewsymbol[
  description={Set of codewords for training at level l}
]{SetCodewordsTraining}{\ensuremath{\mathcal{M}_{t,l}}}
\glsxtrnewsymbol[
  description={Set of codewords for training at level 1}
]{SetCodewordsTrainingInitial}{\ensuremath{\mathcal{M}_{t,1}}}
\glsxtrnewsymbol[
  description={Number of codewords}
]{NumCodewords}{\ensuremath{M}}
\glsxtrnewsymbol[
  description={Number of codewords cubed}
]{NumCodewords3}{\ensuremath{M^3}}
\glsxtrnewsymbol[
  description={Number of codewords for one training iteration}
]{NumCodewordsTraining}{\ensuremath{M_{t,l}}}
\glsxtrnewsymbol[
  description={Number of codewords for one training iteration at level 0.}
]{NumCodewordsTrainingZero}{\ensuremath{M_{t,1}}}
\glsxtrnewsymbol[
  description={Codebook size at level \(l\).}
]{SizeCodebookAtLevel}{\ensuremath{M_l}}
\glsxtrnewsymbol[
  description={Codebook size at level \(L-1\).}
]{SizeCodebookAtLevelMax}{\ensuremath{M_{L}}}
\glsxtrnewsymbol[
  description={Codebook size at level \(l\) for \(x\) axis.}
]{SizeCodebookAtLevelX}{\ensuremath{M_{l,x}}}
\glsxtrnewsymbol[
  description={Codebook size at level \(l\) for \(y\) axis.}
]{SizeCodebookAtLevelY}{\ensuremath{M_{l,y}}}
\glsxtrnewsymbol[
  description={Codebook size at level \(L\) for \(x\) axis.}
]{SizeCodebookAtLevelMaxX}{\ensuremath{M_{L,x}}}
\glsxtrnewsymbol[
  description={Codebook size at level \(L\) for \(y\) axis.}
]{SizeCodebookAtLevelMaxY}{\ensuremath{M_{L,y}}}
\glsxtrnewsymbol[
  description={Codebook size at level \(1\).}
]{SizeCodebookAtZero}{\ensuremath{M_1}}
\glsxtrnewsymbol[
  description={Codebook size at level \(2\).}
]{SizeCodebookAtOne}{\ensuremath{M_2}}
\glsxtrnewsymbol[
  description={Codebook size at level \(3\).}
]{SizeCodebookAtTwo}{\ensuremath{M_3}}
\glsxtrnewsymbol[
  description={Codebook size at level \(4\).}
]{SizeCodebookAtThree}{\ensuremath{M_4}}
\glsxtrnewsymbol[
  description={Number of codewords for one training iteration}
  ]{NumCodewordsTrainingGeneral}{\ensuremath{M_t^\textsf{x}}}
\glsxtrnewsymbol[
  description={Number of codewords for one training iteration}
  ]{NumCodewordsTrainingFull}{\ensuremath{M_t^\mathrm{FS}}}
\glsxtrnewsymbol[
  description={Number of codewords for one training iteration}
  ]{NumCodewordsTrainingHierarchical}{\ensuremath{M_t^\mathrm{HS}}}
\glsxtrnewsymbol[
  description={Number of codewords for one training iteration}
  ]{NumCodewordsTrainingTracking}{\ensuremath{M_t^\mathrm{TS}}}
\glsxtrnewsymbol[
  description={Number of codewords (x axis)}
]{NumCodewordsX}{\ensuremath{M_x}}
\glsxtrnewsymbol[
  description={Number of codewords (y axis)}
]{NumCodewordsY}{\ensuremath{M_y}}
\glsxtrnewsymbol[
  description={BS-RIS distance}
]{DistanceIncident}{\ensuremath{d_i}}
\glsxtrnewsymbol[
  description={Distance}
]{Distance}{\ensuremath{d}}
\glsxtrnewsymbol[
  description={RIS-user distance}
]{DistanceReflected}{\ensuremath{d_r}}
\glsxtrnewsymbol[
  description={RIS-user distance for \(m\) and \(l\)}
]{DistanceReflectedForCodeword}{\ensuremath{d_{m,l}}}
\glsxtrnewsymbol[
  description={RIS-user distance for \(m\)}
]{DistanceReflectedForCodebook}{\ensuremath{d_{m}}}
\glsxtrnewsymbol[
  description={Maxmimum RIS-user distance}
]{DistanceReflectedMax}{\ensuremath{\bar{d}_r}}
\glsxtrnewsymbol[
  description={Gain of transmit antenna}
]{GainTransmitter}{\ensuremath{G_{\mathrm{TX}}}}
\glsxtrnewsymbol[
  description={Gain of receive antenna}
]{GainReceiver}{\ensuremath{G_\mathrm{RX}}}
\glsxtrnewsymbol[
  description={Transmit power}
  ]{PowerTransmitter}{\ensuremath{P}}
\glsxtrnewsymbol[
  description={Effective transmit power}
  ]{PowerTransmitterEffective}{\ensuremath{\bar{P}}}
\glsxtrnewsymbol[
  description={Pathloss, general}
]{Pathloss}{\ensuremath{\mathrm{PL}_s}}
\glsxtrnewsymbol[
  description={Pathloss of incident wave}
]{PathlossIncident}{\ensuremath{\mathrm{PL}_i}}
\glsxtrnewsymbol[
  description={Pathloss of reflected wave}
  ]{PathlossReflected}{\ensuremath{\mathrm{PL}_r}}
\glsxtrnewsymbol[
  description={Power of incident wave}
  ]{PowerWaveIncident}{\ensuremath{P_i}}
\glsxtrnewsymbol[
  description={Power density from incident wave}
]{PowerDensityIncident}{\ensuremath{S_i}}
\glsxtrnewsymbol[
  description={Power density from reflected wave}
]{PowerDensityReflected}{\ensuremath{S_{r,l}}}
\glsxtrnewsymbol[
  description={Power density of reflected wave}
  ]{PowerDensityWaveReflected}{\ensuremath{S}}
\glsxtrnewsymbol[
  description={Distance between RIS and user}
  ]{DistanceRisUser}{\ensuremath{r}}
\glsxtrnewsymbol[
  description={Solid angle of reflected beam}
]{SolidAngleBeamReflected}{\ensuremath{\Omega}}
\glsxtrnewsymbol[
  description={AoD at Bs}
]{AodBs}{\ensuremath{\omega}}
\glsxtrnewsymbol[
  description={Received SNR}
  ]{SnrReceived}{\ensuremath{\gamma}}
\glsxtrnewsymbol[
  description={Received training SNR}
]{SnrReceivedTraining}{\ensuremath{\gamma_{m}}}
\glsxtrnewsymbol[
  description={Received training SNR}
]{SnrReceivedTrainingBoundAtLevel}{\ensuremath{\gamma_{l}}}
\glsxtrnewsymbol[
  description={Received training SNR}
]{SnrReceivedTrainingAtLevel}{\ensuremath{\gamma_{m,l}}}
\glsxtrnewsymbol[
  description={Received training SNR}
]{SnrReceivedAtLevel}{\ensuremath{\bar{\gamma}_{l}}}
\glsxtrnewsymbol[
  description={Received data SNR}
]{SnrReceivedData}{\ensuremath{\gamma^{\star}}}
\glsxtrnewsymbol[
  description={Received data SNR}
]{SnrReceivedDataHierarchical}{\ensuremath{\gamma_\mathrm{HS}^{\star}}}
\glsxtrnewsymbol[
  description={Received data SNR}
]{SnrReceivedDataFull}{\ensuremath{\gamma_\mathrm{FS}^{\star}}}
\glsxtrnewsymbol[
  description={Received data SNR codebook bound}
]{SnrReceivedDataCodebookBound}{\ensuremath{\gamma_\mathrm{opt}^{\star}}}
\glsxtrnewsymbol[
  description={Minimum Received SNR}
  ]{SnrReceivedMin}{\ensuremath{\gamma_\mathrm{min}}}
\glsxtrnewsymbol[
  description={Normalized received training SNR}
]{SnrReceivedNormalized}{\ensuremath{\bar{\gamma}_{m,l}}}
\glsxtrnewsymbol[
  description={Normalized received training SNR}
]{SnrReceivedNormalizedAtLevel}{\ensuremath{\bar{\gamma}_{l}}}
\glsxtrnewsymbol[
  description={Normalized received training SNR}
]{SnrReceivedNormalizedLowest}{\ensuremath{\bar{\gamma}}}
\glsxtrnewsymbol[
  description={Normalized received training SNR}
]{SnrReceivedNormalizedBest}{\ensuremath{\bar{\gamma}_{m_{L}^{\star},L}}}
\glsxtrnewsymbol[
  description={Normalized Received SNR}
]{SnrReceivedFactor}{\ensuremath{\tilde{\gamma}_{m,l}}}
\glsxtrnewsymbol[
  description={Normalized Received SNR}
]{SnrReceivedFactorBound}{\ensuremath{\tilde{\gamma}}}
\glsxtrnewsymbol[
  description={Normalized Received SNR for WBR}
]{SnrReceivedFactorBoundWbr}{\ensuremath{\tilde{\gamma}^\mathrm{WBR}}}
\glsxtrnewsymbol[
  description={Normalized Received SNR for NBR}
]{SnrReceivedFactorBoundNbr}{\ensuremath{\tilde{\gamma}^\mathrm{NBR}}}
\glsxtrnewsymbol[
  description={Total overhead}
  ]{Overhead}{\ensuremath{\tau}}
\glsxtrnewsymbol[
  description={Total overhead, factor, full search}
]{OverheadFactorFull}{\ensuremath{\tilde{\tau}^\mathrm{FS}}}
\glsxtrnewsymbol[
  description={Total overhead, factor, full search}
]{OverheadFactorHierarchical}{\ensuremath{\tilde{\tau}_l^\mathrm{HS}}}
\glsxtrnewsymbol[
  description={Total overhead, factor, full search}
]{OverheadFactorTracking}{\ensuremath{\tilde{\tau}^\mathrm{TS}}}
\glsxtrnewsymbol[
  description={Total overhead, normalized}
]{OverheadRelative}{\ensuremath{\bar{\tau}}}
\glsxtrnewsymbol[
  description={Total overhead, normalized}
]{OverheadRelativeX}{\ensuremath{\bar{\tau}^\textsf{x}}}
\glsxtrnewsymbol[
  description={Total overhead, normalized, bound}
]{OverheadRelativeBound}{\ensuremath{\bar{\tau}^{\prime}}}
\glsxtrnewsymbol[
  description={Total overhead, normalized, bounded bound}
]{OverheadRelativeBoundedBound}{\ensuremath{\bar{\tau}^{\prime\prime}}}
\glsxtrnewsymbol[
  description={Total overhead of full search, normalized}
  ]{OverheadRelativeFull}{\ensuremath{\bar{\tau}^\mathrm{FS}}}
\glsxtrnewsymbol[
  description={Total overhead of hierarchical search, normalized}
  ]{OverheadRelativeHierarchical}{\ensuremath{\bar{\tau}^\mathrm{HS}}}
\glsxtrnewsymbol[
  description={Total overhead of tracking search, normalized}
  ]{OverheadRelativeTracking}{\ensuremath{\bar{\tau}^\mathrm{TS}}}
\glsxtrnewsymbol[
  description={Beam training overhead}
]{OverheadTraining}{\ensuremath{\tau_t}}
\glsxtrnewsymbol[
  description={Beam training overhead of full search}
]{OverheadTrainingFull}{\ensuremath{\tau_t^\mathrm{FS}}}
\glsxtrnewsymbol[
  description={Beam training overhead of full search}
]{OverheadTrainingHierarchical}{\ensuremath{\tau_t^\mathrm{HS}}}
\glsxtrnewsymbol[
  description={Beam training overhead of full search}
]{OverheadTrainingTracking}{\ensuremath{\tau_t^\mathrm{TS}}}
\glsxtrnewsymbol[
  description={Beam training overhead of full search}
]{OverheadTrainingX}{\ensuremath{\tau_t^\textsf{x}}}
\glsxtrnewsymbol[
  description={Beam training overhead of full search}
]{OverheadTrainingY}{\ensuremath{\tau_t^\textsf{y}}}
\glsxtrnewsymbol[
  description={Beam training overhead bound}
]{OverheadTrainingBound}{\ensuremath{\tau_t^\mathrm{max}}}
\glsxtrnewsymbol[
  description={Channel estimation overhead}
  ]{OverheadEstimation}{\ensuremath{\tau_e}}
\glsxtrnewsymbol[
  description={Normalized channel estimation overhead}
  ]{OverheadEstimationNormalized}{\ensuremath{\bar{\tau}_e}}
\glsxtrnewsymbol[
  description={Variable for notational simplification.}
  ]{SimplificationOverhead}{\ensuremath{\alpha}}
\glsxtrnewsymbol[
  description={Variable for notational simplification.}
]{SimplificationOverhead2nd}{\ensuremath{\beta}}
\glsxtrnewsymbol[
  description={Variable for notational simplification for WBR.}
]{SimplificationOverhead2ndWbr}{\ensuremath{\beta^\mathrm{WBR}}}
\glsxtrnewsymbol[
  description={Variable for notational simplification for NBR.}
]{SimplificationOverhead2ndNbr}{\ensuremath{\beta^\mathrm{NBR}}}
\glsxtrnewsymbol[
  description={Frame duration}
  ]{DurationFrame}{\ensuremath{T_f}}
\glsxtrnewsymbol[
  description={Subframe duration}
  ]{DurationSubframe}{\ensuremath{T_\mathrm{sf}}}
\glsxtrnewsymbol[
  description={Symbol duration }
  ]{DurationSymbol}{\ensuremath{T_s}}
\glsxtrnewsymbol[
  description={Symbol duration of training pilots}
  ]{DurationSymbolTraining}{\ensuremath{T_{s}}}
\glsxtrnewsymbol[
  description={Symbol duration of data pilots}
  ]{DurationSymbolData}{\ensuremath{T_{s,e}}}
\glsxtrnewsymbol[
  description={Duration of one codeword}
  ]{DurationCodeword}{\ensuremath{T_\mathrm{cw}}}
\glsxtrnewsymbol[
  description={Hardware-limited codeword duration}
  ]{DurationHardware}{\ensuremath{T_\mathrm{hw}}}
\glsxtrnewsymbol[
  description={Velocity of mobile user}
]{Velocity}{\ensuremath{v}}
\glsxtrnewsymbol[
  description={Velocity bound of mobile user}
]{VelocityBound}{\ensuremath{v^\mathrm{max}}}
\glsxtrnewsymbol[
  description={Velocity limit of mobile user}
]{VelocityLimit}{\ensuremath{\bar{v}}}
\glsxtrnewsymbol[
  description={Number of hierarchy levels}
  ]{NumLevels}{\ensuremath{L}}
\glsxtrnewsymbol[
  description={Number of initial codewords}
]{NumCodewordsInitial}{\ensuremath{M_1}}
\glsxtrnewsymbol[
  description={Higher-level codeword factor}
]{NumCodewordsFactor}{\ensuremath{\tilde{M}}}
\glsxtrnewsymbol[
  description={Higher-level codeword factor}
]{NumCodewordsFactorLevel}{\ensuremath{\tilde{M}^l}}
\glsxtrnewsymbol[
  description={Higher-level codeword factor}
]{NumCodewordsFactorPowerVar}{\ensuremath{\tilde{M}^{l-1}}}
\glsxtrnewsymbol[
  description={Higher-level codeword factor}
]{NumCodewordsFactorPowerLower}{\ensuremath{\tilde{M}^{L-2}}}
\glsxtrnewsymbol[
  description={Higher-level codeword factor}
]{NumCodewordsFactorPowerHigh}{\ensuremath{\tilde{M}^{L-1}}}
\glsxtrnewsymbol[
  description={Scaling with codeword factor}
]{NumCodewordsFactorScaling}{\ensuremath{\tilde{M}^{L-l}}}
\glsxtrnewsymbol[
  description={Higher-level codeword factor}
  ]{NumCodewordsFactorPowerSup}{\ensuremath{\tilde{M}^L}}
\glsxtrnewsymbol[
  description={Total number of codewords considered for hierarchical search}
  ]{NumCodewordsTotalHierarchical}{\ensuremath{\zeta}}
\glsxtrnewsymbol[
  description={Total number of codewords considered for hierarchical search; normalized}
  ]{NumCodewordsTotalNormalizedHierarchical}{\ensuremath{\xi}}
\glsxtrnewsymbol[
  description={Minimum distance of same beam}
  ]{DistanceBeamMin}{\ensuremath{\bar{d}}}
\glsxtrnewsymbol[
  description={Distance that a user travels during one frame}
  ]{DistanceUserFrame}{\ensuremath{d}}
\glsxtrnewsymbol[
  description={System parameter for frame duration}
  ]{FactorFrameDesign}{\ensuremath{\nu}}
\glsxtrnewsymbol[
  description={Number of pilots symbols for channel estimation}
  ]{NumEstimation}{\ensuremath{N_e}}
\glsxtrnewsymbol[
  description={Carrier frequency}
  ]{FrequencyCarrier}{\ensuremath{f}}
\glsxtrnewsymbol[
  description={Channel coherence time}
  ]{DurationChannelCoherence}{\ensuremath{T_c}}
\glsxtrnewsymbol[
  description={Bandwidth of pilot symbols}
  ]{BandwidthPilot}{\ensuremath{B}}
\glsxtrnewsymbol[
  description={Bandwidth of beam pilot symbols}
  ]{BandwidthPilotTraining}{\ensuremath{B_t}}
\glsxtrnewsymbol[
  description={Bandwidth of data symbols}
  ]{BandwidthPilotData}{\ensuremath{B_d}}
\glsxtrnewsymbol[
  description={Noise figure}
  ]{NoiseFigure}{\ensuremath{F}}
\glsxtrnewsymbol[
  description={Noise power spectral density}
  ]{NoisePowerSpectralDensity}{\ensuremath{N_0}}
\glsxtrnewsymbol[
  description={Codeword set at \(l+1\)th level.}
]{CodewordSetAtLevelNext}{\ensuremath{\mathcal{M}_{l+1}}}
\glsxtrnewsymbol[
  description={Effective achievable rate}
]{RateEffective}{\ensuremath{R}}
\glsxtrnewsymbol[
  description={Effective achievable rate}
]{RateEffectiveFull}{\ensuremath{R^\mathrm{FS}}}
\glsxtrnewsymbol[
  description={Effective achievable rate}
]{RateEffectiveHierarchical}{\ensuremath{R^\mathrm{HS}}}
\glsxtrnewsymbol[
  description={Effective achievable rate}
]{RateEffectiveTracking}{\ensuremath{R^\mathrm{TS}}}
\glsxtrnewsymbol[
  description={Noise power for data transmission}
  ]{NoisePower}{\ensuremath{\sigma^2}}
\glsxtrnewsymbol[
  description={Noise power for data transmission}
  ]{NoisePowerData}{\ensuremath{\sigma_d^2}}
\glsxtrnewsymbol[
  description={Noise power for data transmission}
  ]{NoisePowerTraining}{\ensuremath{\sigma_t^2}}
\glsxtrnewsymbol[
  description={Codeword of a codebook}
  ]{Codeword}{\ensuremath{m}}
\glsxtrnewsymbol[
  description={Precoding vector at the AP}
  ]{PrecodingAp}{\ensuremath{\vt{f}}}
\glsxtrnewsymbol[
  description={End-to-end channel}
  ]{ChannelEndToEnd}{\ensuremath{\vt{h}}}
\glsxtrnewsymbol[
  description={End-to-end channel}
  ]{ChannelEndToEndCodeword}{\ensuremath{\vt{h}_m}}
\glsxtrnewsymbol[
  description={End-to-end channel}
]{ChannelEndToEndBestCodeword}{\ensuremath{\vt{h}_{m_{L}^\star}}}
\glsxtrnewsymbol[
  description={Effective area of RIS}
]{AreaRisEffective}{\ensuremath{A_\mathrm{RIS}}}
\glsxtrnewsymbol[
  description={Effective area of receive antenna}
]{AreaReceiveAntennaEffective}{\ensuremath{A_\mathrm{RX}}}
\glsxtrnewsymbol[
  description={Unit cell length x}
]{LengthUnitCellX}{\ensuremath{d_x}}
\glsxtrnewsymbol[
  description={Unit cell length y}
]{LengthUnitCellY}{\ensuremath{d_y}}
\glsxtrnewsymbol[
  description={Unit cell factor}
]{FactorUnitCell}{\ensuremath{g_\mathrm{UC}}}
\glsxtrnewsymbol[
  description={Unit cell factor}
]{FactorUnitCellVariation}{\ensuremath{\tilde{g}_\mathrm{UC}}}
\glsxtrnewsymbol[
  description={RIS gain (with angles) including unit cell gain}
]{GainRisWithAngles}{\ensuremath{g(\theta_{m,l}, \phi_{m,l})}}
\glsxtrnewsymbol[
  description={RIS gain (with angles) including unit cell gain}
]{GainRisWithAnglesUpper}{\ensuremath{g(\theta_{m,l}, \phi_{m,l})}}
\glsxtrnewsymbol[
  description={Coverage area}
  ]{Area}{\ensuremath{A}}
\glsxtrnewsymbol[
  description={Coverage area}
  ]{AreaY}{\ensuremath{A_y}}
\glsxtrnewsymbol[
  description={Coverage area}
  ]{AreaZ}{\ensuremath{A_z}}
\glsxtrnewsymbol[
  description={Coverage area}
  ]{AreaSideMin}{\ensuremath{A_\mathrm{min}^\mathrm{sub}}}
\glsxtrnewsymbol[
  description={Ratio that relates minium codebook size and minimum area size}
  ]{RatioAreaCodebook}{\ensuremath{\mu}}
\glsxtrnewsymbol[
  description={AoA array response}
  ]{ArrayResponseDepartureAp}{\ensuremath{\vt{d}_\mathrm{BS}}}
\glsxtrnewsymbol[
  description={AoA array response}
  ]{ArrayResponseDepartureRis}{\ensuremath{\vt{d}_\mathrm{RIS}}}
\glsxtrnewsymbol[
  description={AoA array response}
  ]{ArrayResponseArrivalRis}{\ensuremath{\vt{a}_\mathrm{RIS}}}
\glsxtrnewsymbol[
  description={Antenna gain transmitter}
  ]{AntennaGainTransmitter}{\ensuremath{D_\mathrm{tx}}}
\glsxtrnewsymbol[
  description={Antenna gain receiver}
  ]{AntennaGainReceiver}{\ensuremath{D_\mathrm{rx}}}
\glsxtrnewsymbol[
  description={Number of scatterers, general}
]{NumScatterers}{\ensuremath{C_s}}
\glsxtrnewsymbol[
  description={Number of scatterers in incident channel}
]{NumScatterersInc}{\ensuremath{C_i}}
\glsxtrnewsymbol[
  description={Number of scatterers in reflected channel}
]{NumScatterersRef}{\ensuremath{C_r}}
\glsxtrnewsymbol[
  description={Incident elevation angle}
]{AngleIncidentElevation}{\ensuremath{\theta_i}}
\glsxtrnewsymbol[
  description={Reflected elevation angle}
]{AngleReflectedElevation}{\ensuremath{\theta_r}}
\glsxtrnewsymbol[
  description={Reflected elevation angle for focus}
]{AngleFocusReflectedElevation}{\ensuremath{\theta_{m,l}}}
\glsxtrnewsymbol[
  description={Incident azimuth angle}
]{AngleIncidentAzimuth}{\ensuremath{\phi_i}}
\glsxtrnewsymbol[
  description={Reflected azimuth angle}
]{AngleReflectedAzimuth}{\ensuremath{\phi_r}}
\glsxtrnewsymbol[
  description={Reflected azimuth angle for focus}
]{AngleFocusReflectedAzimuth}{\ensuremath{\phi_{m,l}}}
\glsxtrnewsymbol[
  description={Shortest path through beam footprint}
]{PathThroughFootprint}{\ensuremath{r_m^\mathrm{min}}}
\glsxtrnewsymbol[
  description={Shortest path through subarea}
]{PathThroughSubarea}{\ensuremath{s_m}}
\glsxtrnewsymbol[
  description={Maximum of shortest path through subarea}
]{PathThroughSubareaMax}{\ensuremath{s_{\mathrm{max}}}}
\glsxtrnewsymbol[
  description={Minimum of shortest path through subarea}
]{PathThroughSubareaMin}{\ensuremath{s_{\mathrm{min}}}}
\glsxtrnewsymbol[
  description={RIS response time}
]{TimeResponse}{\ensuremath{T_r}}
\glsxtrnewsymbol[
  description={Feedback delay}
]{DelayFeedback}{\ensuremath{T_d}}
\glsxtrnewsymbol[
  description={Bound of feedback delay}
]{DelayFeedbackBound}{\ensuremath{T_d^{\mathrm{min}}}}
\glsxtrnewsymbol[
  description={Selected codewords at \(l\)th level}
]{CodewordBestAtLevel}{\ensuremath{m_l^{\star}}}
\glsxtrnewsymbol[
  description={Selected codewords at \(l\)th level}
]{CodewordBestAtPreviousLevel}{\ensuremath{m_{l-1}^{\star}}}
\glsxtrnewsymbol[
  description={Selected codewords at highest level}
]{CodewordBest}{\ensuremath{m_{L}^{\star}}}
\glsxtrnewsymbol[
  description={Channel coefficient of \(c\)th NLoS paths.}
]{ChannelNlos}{\ensuremath{\tilde{h}_{s,c}}}
\glsxtrnewsymbol[
  description={Channel coefficient of \(c\)th NLoS paths.}
]{ChannelNlosIncident}{\ensuremath{\tilde{h}_{i,c}}}
\glsxtrnewsymbol[
  description={Channel coefficient of \(c\)th NLoS paths.}
]{ChannelNlosReflected}{\ensuremath{\tilde{h}_{r,c}}}
\glsxtrnewsymbol[
  description={Wavelength}
]{Wavelength}{\ensuremath{\lambda}}
\glsxtrnewsymbol[
  description={Rician factor, general}
]{FactorRician}{\ensuremath{K_s}}
\glsxtrnewsymbol[
  description={Rician factor, incident}
]{FactorRicianIncident}{\ensuremath{K_i}}
\glsxtrnewsymbol[
  description={Rician factor, reflected}
]{FactorRicianReflected}{\ensuremath{K_r}}
\glsxtrnewsymbol[
  description={Phase of LOS path, general}
]{PhaseLos}{\ensuremath{\varphi_s}}
\glsxtrnewsymbol[
  description={Phase of LOS path, incident}
]{PhaseLosIncident}{\ensuremath{\varphi_i}}
\glsxtrnewsymbol[
  description={Phase of LOS path, reflected}
]{PhaseLosReflected}{\ensuremath{\varphi_r}}
\glsxtrnewsymbol[
  description={Angle dependent RIS area}
]{AreaRisFactor}{\ensuremath{\tilde{a}_\mathrm{RIS}}}
\glsxtrnewsymbol[
  description={Angular width of a beam}
]{AngularWidthBeam}{\ensuremath{\Delta\theta^\mathrm{min}}}
\glsxtrnewsymbol[
  description={First parameter of the overhead parabola.}
]{ParabolaParameterOne}{\ensuremath{a}}
\glsxtrnewsymbol[
  description={First parameter of the overhead parabola, as a bound.}
]{ParabolaParameterOneBound}{\ensuremath{a^{\prime}}}
\glsxtrnewsymbol[
  description={Second parameter of the overhead parabola.}
]{ParabolaParameterTwo}{\ensuremath{b}}
\glsxtrnewsymbol[
  description={Threshold for negligible overhead.}
]{OverheadThreshold}{\ensuremath{\epsilon}}
\glsxtrnewsymbol[
  description={Bound for threshold for negligible overhead.}
]{OverheadThresholdBound}{\ensuremath{\epsilon^\mathrm{max}}}

\newacronym{icd}{ICD}{independent codeword duration}
\newacronym{dcd}{DCD}{dependent codeword duration}
\newacronym{scb}{SCB}{small codebook}
\newacronym{lcb}{LCB}{large codebook}
\newacronym{nbr}{NBR}{narrow-beam regime}
\newacronym{wbr}{WBR}{wide-beam regime}
\newacronym{ior}{SPR}{single pilot symbol regime}
\newacronym{dor}{MPR}{multiple pilot symbol regime}
\newacronym{fs}{FS}{full search}
\newacronym{hs}{HS}{hierarchical search}
\newacronym{ts}{TS}{tracking-based search}
\newacronym{mmWave}{mmWave}{millimeter wave}
\newacronym{rhs}{RHS}{right-hand side}

\begin{document}

\title{Performance Tradeoff Between Overhead and Achievable SNR in RIS Beam Training}

\author{Friedemann~Laue,~\IEEEmembership{Graduate~Student~Member,~IEEE}, Vahid~Jamali,~\IEEEmembership{Member,~IEEE}, and Robert~Schober,~\IEEEmembership{Fellow,~IEEE}
\thanks{This work was presented in part at~\cite{laue2022performancetradeoffris} \textit{(Corresponding author: Friedemann Laue.)}}
\thanks{V. Jamali's work is supported in part by the German Research Foundation (DFG) under CRC MAKI and in part by the LOEWE initiative (Hesse, Germany) within the emergenCITY center.}
\thanks{F. Laue and R. Schober are with Institute for Digital Communications, Friedrich-Alexander-Universität Erlangen-Nürnberg (FAU), 91058 Erlangen, Germany (e-mail: friedemann.laue@fau.de, robert.schober@fau.de).}%
\thanks{F. Laue is also with Fraunhofer IIS, Fraunhofer Institute for Integrated Circuits IIS, Division Communication Systems (e-mail: friedemann.laue@iis.fraunhofer.de).}%
\thanks{V. Jamali is with Department of Electrical Engineering and Information Technology, Technical University of Darmstadt, 64283 Darmstadt, Germany (e-mail: vahid.jamali@tu-darmstadt.de).}%
}

\maketitle

\begin{abstract}
  Efficient beam training is the key challenge in the codebook-based configuration of \glspl{ris} because the beam training overhead can have a strong impact on the achievable system performance.
  In this paper, we study the performance tradeoff between overhead and achievable \gls{snr} in \gls{ris} beam training while taking into account the size of the targeted coverage area, the \gls{ris} response time, and the delay for feedback transmissions.
  Thereby, we consider three common beam training strategies: \gls{fs}, \gls{hs}, and \gls{ts}.
  Our analysis shows that the codebook-based illumination of a given coverage area can be realized with wide- or narrow-beam designs, which result in two different scaling laws for the achievable \gls{snr}.
  Similarly, there are two regimes for the overhead, where the number of pilot symbols required for reliable beam training is dependent on and independent of the \gls{snr}, respectively.
  Based on these insights, we investigate the impact of the beam training overhead on the effective rate and provide an upper bound on the user velocity for which the overhead is negligible.
  Moreover, when the overhead is not negligible, we show that \gls{ts} beam training achieves higher effective rates than \gls{hs} and \gls{fs} beam training, while \gls{hs} beam training may or may not outperform \gls{fs} beam training, depending on the \gls{ris} response time, feedback delay, and codebook size.
  Finally, we present numerical simulation results that verify our theoretical analysis.
  In particular, our results confirm the existence of the proposed regimes, reveal that fast \glspl{ris} can lead to negligible overhead for \gls{fs} beam training, and show that large feedback delays can significantly reduce the performance for \gls{hs} beam training.
\end{abstract}


\glsresetall

\section{Introduction}
\noindent \Glspl{ris} are a promising candidate technology for future wireless networks as they enable smartly configurable communication channels~\cite{renzo2020smartradioenvironments,gong2020smartwirelesscommunications}.
\Glspl{ris} are nearly passive and low-cost devices comprising a large number of unit cells, each imposing an individual phase shift on an incident electromagnetic wave while reflecting it.
By joint configuration of all unit-cell phase shifts, the characteristics of the \gls{ris}, e.g., the angle of reflection, can be controlled.
Hence, tunable \glspl{ris} create new opportunities for network design and are able to improve the system performance of next-generation wireless networks~\cite{wu2021intelligentreflectingsurface}.

In order to maximize the system performance of \gls{ris}-assisted communication networks, the \gls{ris} phase shifts need to be optimized according to the instantaneous channel conditions.
To this end, several frameworks for both phase-shift optimization and channel estimation have been proposed and recent surveys of such techniques can be found in~\cite{pan2022overviewsignalprocessing,zheng2022surveychannelestimation}.
However, the complexity of the corresponding algorithms proposed in the literature typically scales with the number of \gls{ris} unit cells, which may lead to a prohibitively high computational burden for practical deployments~\cite{an2022codebookbasedsolutions}.
In fact, it has been shown that thousands of \gls{ris} unit cells may be required to create a sufficiently strong reflected link for typical outdoor communication use cases~\cite{najafi2020physicsbasedmodeling,liu2022simulationfieldtrial}.
Therefore, low-complexity online algorithms for channel estimation and phase-shift optimization are key challenges for the deployment of large \glspl{ris}.

\subsection{Codebook-Based RIS Configuration}
\noindent One approach for reduction of the \gls{ris} configuration overhead is configuring the unit cells based on a predefined codebook, which is particularly suitable for \gls{los}-dominated channels and high carrier frequencies~\cite{an2022codebookbasedsolutions,pan2022overviewsignalprocessing,peng2022channelestimationris,zheng2022surveychannelestimation}.
The advantage of codebook-based \gls{ris} configuration is twofold.
First, the phase shifts for each codeword can be computed prior to deployment, which alleviates the potentially high complexity of online phase-shift design.
Second, explicit \gls{csi} is not required~\cite{an2022codebookbasedsolutions,wang2023hierarchicalcodebookbased}.
Instead, the codeword providing the largest signal power at the receiver is selected via beam training.
As a result, the overhead incurred by \gls{ris} configuration scales with the size of the codebook, which can be significantly smaller than the number of \gls{ris} unit cells~\cite{najafi2020physicsbasedmodeling,jamali2022lowzerooverhead}.
Moreover, for data transmission, the end-to-end channel resulting for a given codeword can be acquired adopting conventional channel estimation methods~\cite{an2022codebookbasedsolutions}.

\subsection{Performance Tradeoff of RIS Beam Training}
\noindent
Generally speaking, \gls{ris} beam training can be considered as a form of implicit \gls{csi} acquisition, and there exists a fundamental tradeoff between the overhead afforded for \gls{csi} acquisition and the achievable \gls{snr}~\cite{wang2023hierarchicalcodebookbased,an2022codebookbasedsolutions,zhang2022dualcodebookdesign,zhang2023rateoverheadtradeoff,zhang2022rateoverheadtradeoff}.
In other words, a larger amount of time spent for beam training results in a larger achievable \gls{snr} but, at the same time, reduces the time remaining for data transmission.
For codebook-based \gls{ris} configuration, this tradeoff depends on the number of considered codewords, the phase-shift design of each codeword, the applied beam training strategy, and the velocity of the mobile users.

Moreover, there are additional restrictions that need to be considered for analyzing the performance tradeoff of \gls{ris} beam training.
In particular, a communication system is usually subject to time constraints caused by hardware limitations and the employed communication protocol.
For example, a \gls{ris} may be deployed in a standardized system like 3GPP \gls{nr}, which prescribes a specific frame structure that inherently imposes a lower bound on the communication delay.
Although low-latency applications have been targeted in recent \gls{nr} releases, the user-plane and control-plane latencies are \qty{1}{\milli\second} and \qty{10}{\milli\second}, respectively; a reduction by a factor of ten is pursued for future releases~\cite{hamidisepehr20215gurllcevolution,ji2021severalkeytechnologies}.
Nevertheless, this restriction on the communication delay can have an impact on the overhead of \gls{ris} beam training because beam training requires one or multiple transmissions of feedback information.
Furthermore, hardware constraints can influence the achievable system performance because the adopted \gls{ris} technology and the \gls{ris} hardware design have an impact on the \gls{ris} response time, i.e., the time duration required for switching between two stable \gls{ris} configurations.
Table~\ref{tab:technology} summarizes typical response times for current \gls{ris} technologies.
For example, fast switching unit cells can be realized using PIN diodes, whereas liquid crystals increase the \gls{ris} response time by several orders of magnitude.
However, liquid crystal technology is a promising solution for large-scale \glspl{ris} because of the low manufacturing cost and low power consumption~\cite{jimenezsaez2023reconfigurableintelligentsurfaces}.
Thus, a comprehensive analysis of the performance of \gls{ris} beam training should not only consider the properties of the adopted codebook and the employed beam training strategies, but also has to take into account hardware limitations and delay constraints~\cite{wang2022beamtrainingalignment,an2022codebookbasedsolutions}.
\begin{table}
  \begin{center}
    \caption{Response Times for Current \gls{ris} Technologies~\cite{liu2022simulationfieldtrial}.}
    \label{tab:technology}
    \begin{tabular}{| p{6em} | p{7em} | p{13em} | }
      \hline
      Technology & Response Time & Comment\\
      \hline
      PIN diode & \(\sim \unit{\nano\second}\) & Control circuits may reduce response times, \qtyrange{0.33}{28}{\micro\second} reported in~\cite{kamoda201160ghzelectronically,yang20161bit10,li2021programmablemetasurfacebased,pan202110240element}.\\
      \hline
      MEMS & \(\sim \unit{\micro\second}\) & For example, \qty{3}{\micro\second} to \qty{20}{\milli\second} reported in~\cite{sharma2017designsimulationcomparative,sravani2019designperformanceanalysis,nooraiyeen2020designnovellow}.\\
      \hline
      Liquid crystal & \(\sim \unit{\milli\second}\) & \qty{10}{\milli\second} measured for beam steering between \(\pm\qty{60}{\degree}\) \cite{jakoby2020microwaveliquidcrystal}.\\
      \hline
    \end{tabular}
  \end{center}
\end{table}

\subsection{Related Work and Contributions}
\noindent
Recently, beam training and codebook-based \gls{ris} configuration have been studied in several works~\cite{najafi2020physicsbasedmodeling,jamali2021powerefficiencyoverhead,zhang2022rateoverheadtradeoff,jamali2022lowzerooverhead,ghanem2022optimizationbasedphasea,zhang2023rateoverheadtradeoff,zhang2022dualcodebookdesign,wang2023hierarchicalcodebookbased,lv2023risaidedfield,you2020fastbeamtraining,alexandropoulos2022fieldhierarchicalbeama,liu2023lowoverheadbeam,wang2022jointhybrid3d,huang2023twotimescalebased,wang2021jointbeamtraining,tian2021fastbeamtracking}.
Thereby, the main focus was on reducing the beam training overhead compared to a full search of the entire codebook.
A popular approach is to employ hierarchically structured codebooks and multi-lobe beam patterns, which have been studied for both near field and far field deployments~\cite{you2020fastbeamtraining,alexandropoulos2022fieldhierarchicalbeama,zhang2022dualcodebookdesign,wang2023hierarchicalcodebookbased,lv2023risaidedfield}.
Moreover, the beam training overhead was reduced by exploiting deep-learning techniques~\cite{liu2023lowoverheadbeam}, sensing at the \gls{ris}~\cite{wang2022jointhybrid3d}, two time scales of the \gls{ris}-assisted channel~\cite{huang2023twotimescalebased}, and estimation of the \gls{aoa} and \gls{aod} based on random beamforming~\cite{wang2021jointbeamtraining}.
Similarly, prediction and tracking of the user position was investigated in~\cite{tian2021fastbeamtracking}.

However, only few existing works investigated the performance tradeoff between the beam training overhead and the achievable system performance~\cite{zhang2022dualcodebookdesign,zhang2022rateoverheadtradeoff,an2022codebookbasedsolutions,zhang2023rateoverheadtradeoff,wang2023hierarchicalcodebookbased}.
For example, the authors of \cite{zhang2022dualcodebookdesign,zhang2023rateoverheadtradeoff} studied multi-lobe beam training for an intelligent omni-surface and showed that there exists an optimal codebook size maximizing the effective system throughput.
Similar results were obtained for a reconfigurable refractive surface~\cite{zhang2022rateoverheadtradeoff}.
Furthermore, the relation between overhead and the success rate of beam training was studied in~\cite{wang2023hierarchicalcodebookbased}, where non-ideal beam patterns were taken into account.

Nevertheless, some important aspects have not been investigated, yet.
For example, as shown in Table~\ref{tab:technology}, the considered \gls{ris} technology has a strong impact on the \gls{ris} response time, which is an essential parameter of the beam training overhead and has to be taken into account for the tradeoff analysis.
Similarly, the considered feedback delay can be relevant for the beam training overhead.
Moreover, the \gls{ris} codebook is usually designed for covering the entire angular space, but, in practice, a typical deployment aims to illuminate a specific area only~\cite{an2022codebookbasedsolutions}.
Thus, the codebook size can be reduced if the specifics of the target area are taken into account, which inherently reduces the beam training overhead.
However, the relation between codebook size and coverage area has not been considered in the existing studies on \gls{ris} beam training.
In addition, most existing works assume one pilot symbol per codeword during beam training, but, depending on the achievable \gls{ris} gain, additional pilot symbols may be required to select the optimal beam with high probability.

Motivated by the above discussion, this paper provides a comprehensive analysis of codebook-based \gls{ris} configuration and investigates the fundamental performance tradeoff between the beam training overhead and the achievable \gls{snr}.
The main contributions of this paper can be summarized as follows:

\begin{itemize}
  \item We show that codebook-based illumination of a target area by a \gls{ris} can be achieved with narrow- or wide-beam designs, which leads to two different fundamental regimes and scaling laws for the achievable \gls{snr}.
  In addition to the well-known quadratic scaling law for narrow-beam designs, we show that wide-beam designs result in a linear scaling of the \gls{snr} in both the \gls{ris} size and the codebook size.
  \item We show that there are also two regimes for the beam training overhead, characterized by the number of pilot symbols required for reliable beam training.
  In particular, our analysis reveals that the beam training overhead is either dependent on or independent of the \gls{snr}.
  \item We propose a general model for the overhead of multi-level beam training that accounts for the \gls{ris} response time, the feedback delay, and the velocity of a mobile user.
  Based on this model and the proposed regimes, we derive the overhead for reliable beam training for three common beam training strategies: \gls{fs}, \gls{hs}, and \gls{ts}.
  \item Assuming reliable beam training, we show that \gls{ts} beam training achieves higher effective rates than \gls{fs} and \gls{hs} beam training, while we reveal that \gls{hs} beam training may or may not outperform \gls{fs} beam training, depending on the \gls{ris} response time, feedback delay, and codebook size.
  In addition, we provide upper bounds for the mobile user's velocity, which guarantee that the beam training overhead has negligible impact on system performance.
  \item We present a comprehensive set of simulation results that corroborate our theoretical analysis.
  In particular, we verify the existence of the proposed regimes and validate the condition for negligible beam training overhead.
  Moreover, our simulation results reveal that fast \glspl{ris} facilitate \gls{fs} beam training, whereas large feedback delays can significantly reduce the performance for \gls{hs} beam training.
\end{itemize}

The remainder of this paper is organized as follows. Section~\ref{sec:system_model} describes the considered communication system.
The different regimes for the \gls{snr} and the overhead are presented in Section~\ref{sec:regimes}.
A detailed analysis of the considered beam training strategies and the overheads they cause is provided in Section~\ref{sec:beam-training-analysis}.
These results are adopted in Section~\ref{sec:performance-tradeoff} to study the performance tradeoff of \gls{ris} beam training.
Finally, numerical simulation results are presented in Section~\ref{sec:numerical_results}, and conclusions are drawn in Section~\ref{sec:conclusion}.

\textit{Notations}: Bold small and capital letters are used to denote vectors and matrices, respectively. The zero vector of size \(n\) and the identity matrix of size \(n\times n\) are represented by \(\vt{0}_n\) and \(\mt{I}_n\), respectively.
Moreover, \((\cdot)^H\) denotes the Hermitian operator, \(\E{\cdot}\) represents expectation, and \(\abs{\cdot}\) denotes the magnitude of a scalar or the cardinality of a set.
The smallest integer larger than \(x\) is represented by \(\ceil*{x}\).
Furthermore, a complex Gaussian distribution with mean vector \(\vt{m}\) and covariance matrix \(\mt{C}\) is denoted by \(\CN{\vt{m}, \mt{C}}\), and \(\mathcal{U}(a, b)\) denotes a uniform distribution in the interval \([a,b]\).

\section{System Model}\label{sec:system_model}
  \noindent
  As illustrated in Fig.~\ref{fig:system}, we consider a narrow-band \gls{miso} communication system comprising a \gls{bs}, a \gls{ris}, and a mobile user that travels within a given coverage area of size \(\gls{Area}\).
  Adopting a high carrier frequency, e.g., in the \gls{mmWave} band, we assume \gls{los}-dominated channels and a blocked direct link between the \gls{bs} and the user~\cite{wang2023hierarchicalcodebookbased}.
  Thus, the \gls{bs} transmits data to the user via the \gls{ris}, where we assume that both the \gls{bs} and the user are in the far-field of the \gls{ris}.
  The \gls{bs} comprises \(\gls{NumAp}\) antenna elements with element spacing \(d\) and antenna gain \(\gls{GainTransmitter}\), and the user is equipped with a single antenna with gain \(\gls{GainReceiver}\).
  Moreover, we assume that the \gls{ris} and the coverage area are located in the \(x\)-\(y\) and \(y\)-\(z\) planes, respectively.
  At the \gls{ris}, the elevation and azimuth \glspl{aoa} of the incident wave and \glspl{aod} of the reflected wave are denoted by \((\gls{AngleIncidentElevation}, \gls{AngleIncidentAzimuth})\) and \((\gls{AngleReflectedElevation},\gls{AngleReflectedAzimuth})\), respectively.
  Furthermore, we assume a square \gls{ris} comprising \(\gls{NumRis}\) unit cells of lengths \(\gls{LengthUnitCellX}\) and \(\gls{LengthUnitCellY}\) along the \(x\)- and \(y\)-directions, respectively.
  In order to align the \gls{ris} reflection beam with the location of the mobile user, the \gls{ris} is regularly reconfigured by means of beam training.
  Thereby, in general, we assume a predefined multi-level \gls{ris} codebook with \(\gls{NumLevels}\) levels\footnote{In Section~\ref{sec:beam-training-analysis}, we study \gls{fs} beam training and \gls{ts} beam training with \(\gls{NumLevels}=1\) as well as \gls{hs} beam training with \(\gls{NumLevels} > 1\).}.
  The codewords at the \(l\)th level, \(l \in \{1, \dotsc, \gls{NumLevels}\}\), are designed to illuminate a particular subarea of the coverage area, respectively, and the set of codewords at the \(l\)th level is denoted by \(\gls{SetCodewords}\).
  In addition, \(\gls{SetCodewordsTraining} \subseteq \gls{SetCodewords}\) denotes the subset of codewords that are actually used for beam training, which is specified by the considered beam training strategy, see Section~\ref{sec:beam_training_strategies}.
  Finally, the number of codewords in sets \(\gls{SetCodewords}\) and \(\gls{SetCodewordsTraining}\) are denoted by \(\abs*{\gls{SetCodewords}} = \gls{SizeCodebookAtLevel}\) and \(\abs*{\gls{SetCodewordsTraining}} = \gls{NumCodewordsTraining}\), respectively.

  \begin{figure}
    \centering
    \includegraphics{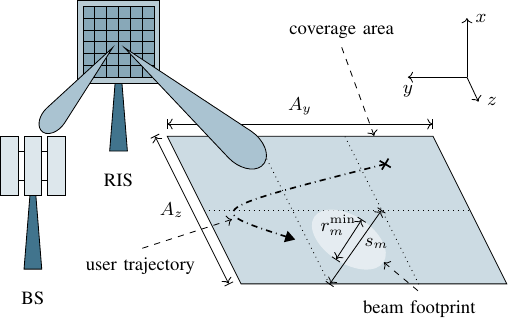}
    \caption{\Glsxtrshort{ris}-assisted \glsxtrshort{miso} system with a mobile user.}
    \label{fig:system}
  \end{figure}

  \subsection{Transmission and Beam Training Protocol}\label{sec:transmission_protocol}
  \noindent
  We consider a frame-based transmission protocol and assume that the frame duration \(\gls{DurationFrame}\) corresponds to the beam coherence time, i.e., the time duration for which the mobile user is covered by the beam of a given codeword.
  Thus, beam training is required once per frame.
  In this work, we assume downlink beam training, where the \gls{bs} transmits pilot symbols to the user, while the \gls{ris} changes\footnote{In practice, beam training requires synchronization between user, \gls{bs}, and \gls{ris}, which may increase the overall overhead.
  However, we neglect the impact of synchronization in this work because synchronization errors below \qty{500}{\nano\second} can be realized with recent releases of 3GPP \gls{nr}~\cite{romanov2021precisesynchronizationmethod,hamidisepehr20215gurllcevolution}, which is smaller than the symbol duration of the considered narrow-band system.} its phase-shift configuration based on the predefined codebook.
  
  \begin{figure}
    \centerline{
      \includegraphics{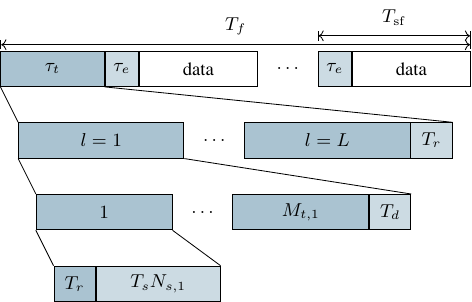}
    }
    \caption{Illustration of one frame of the considered transmission protocol.}
    \label{fig:protocol}
  \end{figure}
  As illustrated in Fig.~\ref{fig:protocol}, the first segment of a frame is determined by the total beam training overhead \(\gls{OverheadTraining}\) resulting from beam training with \(\gls{NumLevels}\) codebook levels.
  In particular, at the \(l\)th level, each codeword in \(\gls{SetCodewordsTraining}\) causes an overhead of \(\gls{TimeResponse} + \gls{DurationSymbolTraining} \gls{NumSymTxLevel}\), where \(\gls{TimeResponse}\), \(\gls{DurationSymbolTraining}\), and \(\gls{NumSymTxLevel}\) denote the \gls{ris} response time (i.e., the time needed by the \gls{ris} to change from one phase-shift configuration to another, see Table~\ref{tab:technology}), the pilot symbol duration, and the number of pilot symbols, respectively.
  In addition, at each codebook level, the user determines the codeword that provides the largest received signal power and transmits the corresponding codeword index \(\gls{CodewordBestAtLevel}\) over a dedicated feedback channel to the \gls{ris}, which introduces a delay of \(\gls{DelayFeedback}\).
  Then, assuming zero processing delay\footnote{We only focus on the systematic overhead of the considered beam training protocol and neglect possible processing overheads. In practice, \(\gls{CodewordSetAtLevelNext}\) may be found based on \(\gls{CodewordBestAtLevel}\) using a simple look-up table.}, \(\gls{CodewordBestAtLevel}\) is used to determine the subset of codewords for the training at level \(l+1\).
  Finally, beam training is concluded by configuring codeword \(\gls{CodewordBest}\) for subsequent data transmission, which adds a further overhead of \(\gls{TimeResponse}\).
  Thus, the total beam training overhead is given by
  \begin{equation}\label{eq:overhead_training_general}
    \gls{OverheadTraining} = \gls{TimeResponse} + \gls{NumLevels}\gls{DelayFeedback} + \sum_{l=1}^{\gls{NumLevels}} \left(\gls{TimeResponse} + \gls{DurationSymbolTraining} \gls{NumSymTxLevel}\right) \gls{NumCodewordsTraining}.
  \end{equation}

  After beam training, the \gls{ris} configuration is fixed to estimate the end-to-end channel for data transmission.
  To this end, we adopt the channel coherence time in~\cite[Eq. (4.40.c)]{rappaport2001wirelesscommunications} to divide the remaining time \(\gls{DurationFrame} - \gls{OverheadTraining}\) into subframes of duration \(\gls{DurationSubframe} = \sqrt{\frac{9}{16\pi}} \frac{\gls{Wavelength}}{\gls{Velocity}}\), where \(\gls{Wavelength}\) and \(\gls{Velocity}\) denote the wavelength and the mobile user's velocity, respectively.
  Thus, assuming a channel estimation overhead of \(\gls{OverheadEstimation}\) in each subframe, the overall overhead for beam training and channel estimation normalized by the frame duration is given by
  \begin{equation}\label{eq:overall-overhead}
    \gls{OverheadRelative} =
    \frac{\gls{OverheadTraining} + \frac{\gls{DurationFrame}-\gls{OverheadTraining}}{\gls{DurationSubframe}}\gls{OverheadEstimation}}{\gls{DurationFrame}}
    = \frac{\gls{OverheadTraining}}{\gls{DurationFrame}} \left(1 - \gls{OverheadEstimationNormalized}\right) + \gls{OverheadEstimationNormalized},
  \end{equation}
  where \(\gls{OverheadEstimationNormalized} = \gls{OverheadEstimation}/\gls{DurationSubframe} = \gls{OverheadEstimation} 4 \sqrt{\pi} \gls{Velocity}/(3 \gls{Wavelength})\) denotes the normalized estimation overhead, and we assume that \(\frac{\gls{DurationFrame}-\gls{OverheadTraining}}{\gls{DurationSubframe}}\) is an integer.

  \subsection{Signal Model}
  \noindent
  During beam training, the received signal for the \(m\)th codeword is given by
  \begin{equation}\label{eq:signal-model}
    \vt{y}_m = \sqrt{\gls{GainReceiver} \gls{GainTransmitter} \gls{PowerTransmitter}} \glsConjTrans{ChannelEndToEndCodeword} \gls{PrecodingAp} \gls{TransmitSymbolVector} + \vt{n},
  \end{equation}
  where \(\gls{TransmitSymbolVector} \in \mathbb{C}^{\gls{NumSymTxLevel}}\), \(\gls{ChannelEndToEnd}_m \in \mathbb{C}^{\gls{NumAp}}\),
  \(\gls{PrecodingAp} \in \mathbb{C}^{\gls{NumAp}}\), \(\gls{PowerTransmitter}\), and \(\vt{n} \sim\CN{\vt{0}_{\gls{NumSymTxLevel}}, \gls{NoisePower} \mt{I}_{\gls{NumSymTxLevel}}}\) denote the transmit vector comprising \(\gls{NumSymTxLevel}\) pilot symbols of unit power, the end-to-end \gls{bs}-\gls{ris}-user channel for the \(m\)th codeword, the precoding vector at the \gls{bs}, the transmit power, and zero-mean Gaussian noise with power \(\gls{NoisePower}\), respectively.
  The precoding vector \(\gls{PrecodingAp}\) is chosen to align the main lobe of the \gls{bs} with the \gls{los} path to the \gls{ris}~\cite{wang2023hierarchicalcodebookbased}, and the end-to-end channel \(\gls{ChannelEndToEndCodeword}\) is given by\footnote{Since the direct \gls{bs}-user link is assumed to be blocked, its contribution to \(\gls{ChannelEndToEndCodeword}\) is neglected. Thus, \eqref{eq:channel-end-to-end} only comprises the \gls{bs}-\gls{ris}-user link.}
  \begin{equation}\label{eq:channel-end-to-end}
    \glsConjTrans{ChannelEndToEndCodeword} = \vt{h}_r^H \mt{G}_m \mt{H}_i,
  \end{equation}
  where \(\mt{H}_i \in \mathbb{C}^{\gls{NumRis}\times \gls{NumAp}}\),
  \(\mt{G}_m \in \mathbb{C}^{\gls{NumRis}\times \gls{NumRis}}\),
  and \(\vt{h}_r \in \mathbb{C}^{\gls{NumRis}}\) denote the \gls{bs}-\gls{ris} channel matrix, the reflection matrix of the \gls{ris}, and the \gls{ris}-user channel vector, respectively.
  For \(\mt{H}_i\) and \(\vt{h}_r\), we assume a statistical channel model\footnote{For example, we adopt a geometric Rician fading channel model for our numerical simulations, see Section~\ref{sec:numerical_results}.} with a dominant \gls{los} path.
  Moreover, the \gls{ris} reflection matrix is given by
  \(
    \mt{G}_m = \diag*{
    \begin{bmatrix}
      \gls{FactorUnitCell} e^{j\gls{PhaseUc}_{m,1}} & \gls{FactorUnitCell} e^{j\gls{PhaseUc}_{m,2}} & \dots & \gls{FactorUnitCell} e^{j\gls{PhaseUc}_{m,\gls{NumRis}}}
      \end{bmatrix}
    }
  \),
  where \(\gls{FactorUnitCell}\) and $\psi_{m,q}$ denote the unit-cell factor and the phase shift of the \(q\)th unit cell, respectively.
  The unit cell factor is given by \(\gls{FactorUnitCell} = \frac{\mathsf{j} 4\pi \gls{LengthUnitCellX}\gls{LengthUnitCellY}}{\glsSquared{Wavelength}}\gls{FactorUnitCellVariation}\), where \(\gls{FactorUnitCellVariation}\) accounts for the \gls{ris} \glspl{aoa} and \glspl{aod} and the polarization of the incident wave, see~\cite{najafi2020physicsbasedmodeling} for details.
  
  \subsection{System Performance Metrics}
  \noindent
  For data transmission, it is desired that the selected codeword maximizes the end-to-end channel gain \(\abs*{\sqrt{\gls{GainReceiver} \gls{GainTransmitter} \gls{PowerTransmitter}} \glsConjTrans{ChannelEndToEndCodeword} \gls{PrecodingAp}}^2\) in \eqref{eq:signal-model}, which is achieved by matched-filtering the received training signal with the known pilot sequence.
  Hence, the codeword selection at the \(l\)th codebook level is given by
  \begin{equation}\label{eq:selected-codeword}
    \gls{CodewordBestAtLevel} = \argmax_{m\in\gls{SetCodewordsTraining}} \abs*{\glsConjTrans{TransmitSymbolVector} \vt{y}_m}^2.
  \end{equation}
  Based on the filtered training signal, we define the training \gls{snr} as follows
  \begin{equation}\label{eq:snr-training}
    \gls{SnrReceivedTrainingAtLevel} = \frac{\gls{GainReceiver} \gls{GainTransmitter} \gls{PowerTransmitter} \gls{NumSymTxLevel} \abs*{\glsConjTrans{ChannelEndToEndCodeword} \gls{PrecodingAp}}^2}{\gls{NoisePower}} = \gls{NumSymTxLevel}\gls{SnrReceivedNormalized},
  \end{equation}
  where \(\gls{SnrReceivedNormalized} = \gls{GainReceiver} \gls{GainTransmitter} \gls{PowerTransmitter} \abs*{\glsConjTrans{ChannelEndToEndCodeword} \gls{PrecodingAp}}^2 / \gls{NoisePower}\) denotes the achievable \gls{snr} per symbol for the \(m\)th codeword at the \(l\)th codebook level.

  Moreover, we aim to analyze the system performance by evaluating the tradeoff between the achievable \gls{snr} for data transmission and the overall beam training overhead.
  The latter is given by \(\gls{OverheadRelative}\) in~\eqref{eq:overall-overhead}.
  The former, assuming a unit-power data symbol, is given by
  \begin{equation}\label{eq:snr-data}
    \gls{SnrReceivedData} = \gls{SnrReceivedNormalizedBest} =
    \frac{
      \gls{GainReceiver} \gls{GainTransmitter} \gls{PowerTransmitter}
      \abs*{\glsConjTrans{ChannelEndToEndBestCodeword} \gls{PrecodingAp}}^2
      }{
        \gls{NoisePower}
      }.
  \end{equation}
  Thus, the performance tradeoff can be analyzed based on the effective ergodic rate
  \begin{equation}\label{eq:rate_effective}
    \gls{RateEffective}
    = \left(1 - \gls{OverheadRelative}\right)
      \E*{\log_2\left(1 + \gls{SnrReceivedData}\right)}.
  \end{equation}
  In the following sections, we analyze \(\gls{SnrReceivedData}\) and \(\gls{OverheadRelative}\) in detail to reveal their impact on \eqref{eq:rate_effective}.
  
  \begin{remark}\label{rm:impact_on_decisions}
    The probability of selecting the best codeword during beam training depends on several parameters.
    For example, receiver noise and channel fading have an impact on the instantaneous signal power \(\abs*{\glsConjTrans{TransmitSymbolVector} \vt{y}_m}^2\) in \eqref{eq:selected-codeword}, which may lead to a suboptimal codeword selection.
    Similarly, beam training is influenced by the codebook design because it affects the beam patterns and thus the received power for a particular codeword~\cite{wang2023hierarchicalcodebookbased}.
    Moreover, erroneous codeword selection in multi-level beam training propagates from one level to the next, which may result in beam misalignment~\cite{wang2022beamtrainingalignment}.
    Since reliable beam training is desired, we propose two regimes for the beam training overhead in Section~\ref{sec:regimes-overhead}, which guarantee that the training \gls{snr} in \eqref{eq:snr-training} is sufficiently large to ensure selection of the best beam possible.
    Nevertheless, the effect of erroneous codeword selection is included in our numerical results in Section~\ref{sec:numerical_results}.
  \end{remark}

\section{Regimes for SNR and Overhead}\label{sec:regimes}
\noindent This section introduces fundamental regimes for both the achievable \gls{snr} and the beam training overhead, which form the basis for the beam training analysis in Section~\ref{sec:beam-training-analysis} and the tradeoff analysis in Section~\ref{sec:performance-tradeoff}.

\subsection{Achievable SNR}\label{sec:achievable_snr}
\noindent
The achievable \gls{snr} for codebook-based \gls{ris} configuration can be characterized by two fundamental regimes, which depend, among other parameters, on the size of the considered codebook~\cite{laue2022performancetradeoffris}.
In the following, we provide a general definition of these \gls{snr} regimes, which takes the beam width and the size of the coverage area into account\footnote{We only focus on the reflected beam from the \gls{ris} to the user because the incident beam is aligned with the \gls{los} path to the \gls{bs}.}.
To this end, without loss of generality, we first focus on the \(m\)th codeword and the \(m\)th subarea for the \(l\)th codebook level, and introduce the following definitions of narrow beams and wide beams.
For a narrow beam, the phase shifts for the \(m\)th codeword are designed to maximize the received power for particular values of \((\gls{AngleReflectedElevation}, \gls{AngleReflectedAzimuth})\), e.g., for the \glspl{aod} from the \gls{ris} to the center of the \(m\)th subarea.
Such a design yields the minimum possible beam width and results in the smallest diameter of the beam footprint, denoted by \(\gls{PathThroughFootprint}\), see Fig.~\ref{fig:system}.
For a wide beam, in contrast, the phase shifts are designed for non-degenerate intervals of \((\gls{AngleReflectedElevation}, \gls{AngleReflectedAzimuth})\), which leads to beam footprints wider than \(\gls{PathThroughFootprint}\).

\begin{proposition}\label{prop:snr-scaling-law}
  If the \(m\)th codeword at the \(l\)th codebook level is designed to fully illuminate the \(m\)th subarea, the achievable \gls{snr} \(\gls{SnrReceivedNormalized}\) scales as follows
  \begin{equation}\label{eq:snr-scaling-law}
    \gls{SnrReceivedNormalized} \propto
    \begin{cases}
      \gls{NumRis} \gls{SizeCodebookAtLevel}
      & \gls{PathThroughFootprint} \ll \gls{PathThroughSubarea}\\
      \glsSquared{NumRis}
      & \gls{PathThroughFootprint} \gg \gls{PathThroughSubarea},
    \end{cases}
  \end{equation}
  where \(\gls{PathThroughSubarea}\) denotes the diameter of the \(m\)th subarea.
\end{proposition}

\begin{proof}
  If \(\gls{PathThroughFootprint} \ll \gls{PathThroughSubarea}\) holds, the footprint of a narrow beam is not wide enough to fully illuminate the targeted subarea, i.e., a wide-beam design is required to obtain the desired coverage.
  In this work, we adopt an idealized wide-beam design that perfectly reflects the power collected at the \gls{ris} to the targeted subarea, because the \gls{ris} gain for existing wide-beam designs is usually not in tractable analytical form~\cite{laue2021irsassistedactive,jamali2021powerefficiencyoverhead}.
  For the idealized wide-beam design, it is shown in Appendix~\ref{sec:proof-snr-scaling-wbr} that \(\gls{SnrReceivedNormalized} \propto \gls{NumRis} \gls{SizeCodebookAtLevel}\).
  In contrast, if \(\gls{PathThroughFootprint} \gg \gls{PathThroughSubarea}\) holds, a narrow beam is sufficient to fully illuminate the targeted subarea, which results in \(\gls{SnrReceivedNormalized} \propto \glsSquared{NumRis}\), see Appendix~\ref{sec:proof-snr-scaling-nbr}.
\end{proof}

Based on Proposition~\ref{prop:snr-scaling-law}, we define the following regimes for the achievable \gls{snr} for codebook-based \gls{ris} configuration.
The system operates in the \gls{wbr} if \(\gls{PathThroughFootprint} \ll \gls{PathThroughSubarea}\) holds for all codewords \(m \in \gls{SetCodewords}\) at the \(l\)th codebook level.
Then, wide beams are required to fully illuminate the subareas of the coverage area and \(\gls{SnrReceivedNormalized} \propto \gls{NumRis} \gls{SizeCodebookAtLevel}\).
In contrast, the system operates in the \gls{nbr} if \(\gls{PathThroughFootprint} \gg \gls{PathThroughSubarea}\) holds for all codewords \(m \in \gls{SetCodewords}\) at the \(l\)th codebook level.
In this case, narrow beams are sufficient to fully illuminate the subareas of the coverage area and \(\gls{SnrReceivedNormalized} \propto \glsSquared{NumRis}\).
We note that the transition between the \gls{wbr} and the \gls{nbr} is gradual and continuous.
Although this transition regime is not captured in Proposition~\ref{prop:snr-scaling-law}, the \glspl{snr} in the \gls{wbr} and the \gls{nbr} can be considered as lower and upper bounds for the transition regime, respectively.
Therefore, this paper focuses on the \gls{wbr} and \gls{nbr}, and leaves a detailed analysis of the transition regime for future work.

In order to obtain more insight on the \gls{snr} regimes, consider the following assumptions on the coverage area and the codebook.
We assume a rectangular coverage area with center coordinates \((x_A, y_A, z_A)\) and size \(\gls{Area} = \gls{AreaY} \gls{AreaZ}\), where \(\gls{AreaY}\) and \(\gls{AreaZ}\) denote the side lengths in the \(y\)- and \(z\)-directions, respectively.
The coverage area is illuminated by \(\gls{SizeCodebookAtLevel} = \gls{SizeCodebookAtLevelX} \gls{SizeCodebookAtLevelY}\) codewords, i.e., the codebook at the \(l\)th level comprises \(\gls{SizeCodebookAtLevelX}\) and \(\gls{SizeCodebookAtLevelY}\) different phase-shift profiles along the \(x\)- and \(y\)-directions of the \gls{ris}, respectively, indexed by \(m_x\in\{1, \dotsc, \gls{SizeCodebookAtLevelX}\}\) and \(m_y\in\{1, \dotsc, \gls{SizeCodebookAtLevelY}\}\).
This corresponds to partitioning the coverage area into \(\gls{SizeCodebookAtLevel}\) subareas, each of size \((\gls{AreaY}/\gls{SizeCodebookAtLevelY})(\gls{AreaZ}/\gls{SizeCodebookAtLevelX})\).
Furthermore, assume that the beam of the \(m\)th codeword is directed to the center of the \((m_x,m_y)\)th subarea, which is located at distance
\(
  \glsSquared{DistanceReflectedForCodebook}
  = x_A^2
  + \left(y_A + \frac{\gls{AreaY}}{2\gls{SizeCodebookAtLevelY}} \left(2m_y - 1 - \gls{SizeCodebookAtLevelY}\right)\right)^2
  + \left(z_A + \frac{\gls{AreaZ}}{2\gls{SizeCodebookAtLevelX}} \left(2m_x - 1 - \gls{SizeCodebookAtLevelX}\right)\right)^2
\)
from the \gls{ris}.
Then, the minimum diameter of the beam footprint is given by \(\gls{PathThroughFootprint} = 2 \gls{DistanceReflectedForCodebook} \tan(\gls{AngularWidthBeam}/2)\), where \(\gls{AngularWidthBeam}\) denotes the minimum angular beam width of a narrow beam.
The latter, assuming half-wavelength sized unit cells and the \qty{3}{\decibel} beam width in broadside direction, is given by \(\gls{AngularWidthBeam} = |\frac{\pi}{2} - \arccos(\frac{2.782}{\pi \sqrt{\gls{NumRis}}})| + |\frac{\pi}{2} - \arccos(\frac{-2.782}{\pi \sqrt{\gls{NumRis}}})|\)~\cite[Eq.~(27)]{han2021halfpowerbeamwidth}.
Consequently, we obtain
\begin{equation}\label{eq:width-footprint}
  \gls{PathThroughFootprint} = 2 \gls{DistanceReflectedForCodebook} \tan(\gls{AngularWidthBeam}/2) \approx 1.77 \frac{\gls{DistanceReflectedForCodebook}}{\sqrt{\gls{NumRis}}},
\end{equation}
where the approximation is based on \(\tan(x) \approx x\) and \(\arccos(x) \approx \frac{\pi}{2} - x\) for \(\abs*{x} < 0.5\), which can be assumed since usually \(\gls{NumRis} > 13 > 16(2.782/\pi)^2\).
As a result, the \gls{wbr} and \gls{nbr} can be characterized by \(\gls{NumRis}\), \(\gls{Area}\), and \(\gls{SizeCodebookAtLevel}\) as follows.
\begin{lemma}\label{lm:snr-regimes}
  The system operates in the \gls{wbr} if
  \begin{equation}\label{eq:condition_wbr}
    \gls{NumRis} \gg
    \frac{
      \left(1.77 \max_{m \in \gls{SetCodewords}}
      \gls{DistanceReflectedForCodebook}\right)^2
    }{
      \left(\gls{AreaY}/\gls{SizeCodebookAtLevelY}\right)^2
      + \left(\gls{AreaZ}/\gls{SizeCodebookAtLevelX}\right)^2
    }.
  \end{equation}
  On the contrary, the system operates in the \gls{nbr} if
  \begin{equation}\label{eq:condition_nbr}
    \gls{NumRis} \ll
    \frac{
      \left(1.77 \min_{m \in \gls{SetCodewords}}
      \gls{DistanceReflectedForCodebook}\right)^2
    }{
      \left(\gls{AreaY}/\gls{SizeCodebookAtLevelY}\right)^2
      + \left(\gls{AreaZ}/\gls{SizeCodebookAtLevelX}\right)^2
    }.
  \end{equation}
\end{lemma}

\begin{proof}
  Since all subareas have equal size, we have
  \begin{equation}\label{eq:diameter}
    \gls{PathThroughSubarea} = \sqrt{(\gls{AreaY}/\gls{SizeCodebookAtLevelY})^2 + (\gls{AreaZ}/\gls{SizeCodebookAtLevelX})^2}
  \end{equation}
  for all codewords.
  Thus, according to the above definition of the \gls{snr} regimes, the system operates in the \gls{wbr} and \gls{nbr} for \(\max_{m \in \gls{SetCodewords}} \gls{PathThroughFootprint} \ll \gls{PathThroughSubarea}\) and \(\min_{m \in \gls{SetCodewords}} \gls{PathThroughFootprint} \gg \gls{PathThroughSubarea}\), respectively, which can be rewritten as \eqref{eq:condition_wbr} and \eqref{eq:condition_nbr} using \eqref{eq:width-footprint} and \eqref{eq:diameter}.
\end{proof}

The results of Lemma~\ref{lm:snr-regimes} are visualized in Fig.~\ref{fig:footprints_subareas_regimes}, which shows the impact of the codebook size and the \gls{ris} size on the \gls{snr} regimes.
For visualization purpose, we adopt the \glspl{rhs} of \eqref{eq:condition_wbr} and \eqref{eq:condition_nbr} to separate the \gls{wbr} and \gls{nbr}, although there is a continuous transition in practice.
Nevertheless, Fig.~\ref{fig:footprints_subareas_regimes} suggests that the \gls{wbr} is particularly relevant for very large \glspl{ris} or small codebooks, e.g., at the lower levels of a hierarchical codebook.
Moreover, for \(\gls{NumLevels} > 1\), the \gls{snr} regime may change during beam training, which is indicated by the markers in Fig.~\ref{fig:footprints_subareas_regimes} assuming \(\gls{NumRis} = 3600\) and \(\gls{NumLevels} = 5\).
In particular, one can see that for the beam training in levels \(l \leq 3\) the \gls{wbr} is relevant, whereas for \(l = 4\) and \(l = 5\) the transition regime and the \gls{nbr} apply, respectively.
This observation is confirmed by our numerical results in Section~\ref{sec:numerical_results_snr_regimes}, where Fig.~\ref{fig:snr_verification} shows that the \gls{snr} changes from linear scaling (\gls{wbr}) to saturation (\gls{nbr}) as the codebook size grows.
\begin{figure}[!t]
  \centering
  \includegraphics{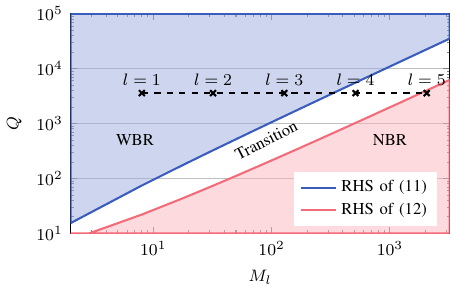}
  \caption{\Gls{snr} regimes \gls{wbr} and \gls{nbr} as functions of \(\gls{SizeCodebookAtLevel}\) and \(\gls{NumRis}\), visualized based on the \glsxtrlongpl{rhs} of \eqref{eq:condition_wbr} and \eqref{eq:condition_nbr} for \((x_A, y_A, z_A) = (\qty{-10}{\metre}, \qty{-20}{\metre}, \qty{100}{\metre})\), and \(\gls{Area} = \qtyproduct{100 x 50}{\metre}\).}
  \label{fig:footprints_subareas_regimes}
\end{figure}

\subsection{Beam Training Overhead}\label{sec:regimes-overhead}
\noindent Next, we show that the beam training overhead is also characterized by two fundamental regimes, in which \(\gls{OverheadTraining}\) is dependent on and independent of the \gls{snr}, respectively.
To this end, recall that the beam training is influenced by receiver noise and channel fading, which may affect the codeword selection in \eqref{eq:selected-codeword}.
Thus, in order to ensure that the selected codeword is correctly aligned with the current user location with high probability, we require that the optimal codeword at each level achieves a minimum \gls{snr} at the user.
As a result, we assume that reliable beam training is possible if \(\max_{m \in \gls{SetCodewordsTraining}} \gls{SnrReceivedTrainingAtLevel} \geq \gls{SnrReceivedMin}\), which by using \eqref{eq:snr-training} can be rewritten as
\begin{equation}\label{eq:pilot_symbol_bound}
  \gls{NumSymTxLevel} \geq
  \frac{\gls{SnrReceivedMin}}{\gls{SnrReceivedNormalizedAtLevel}},
\end{equation}
where \(\gls{SnrReceivedNormalizedAtLevel} = \max_{m \in \gls{SetCodewordsTraining}} \gls{SnrReceivedNormalized}\).
Based on \eqref{eq:pilot_symbol_bound}, we define two overhead regimes as follows.
If \(\gls{NumSymTxLevel} \geq \gls{SnrReceivedMin}/\gls{SnrReceivedNormalizedAtLevel} > 1\), more than one pilot symbol is required for reliable beam training, which we refer to as the \gls{dor}.
In this case, the beam training overhead in \eqref{eq:overhead_training_general} is a function of the \gls{snr}.
Moreover, we set \(\gls{NumSymTxLevel} = 1\) if \(\gls{SnrReceivedMin}/\gls{SnrReceivedNormalizedAtLevel} < 1\) because at least one pilot symbol is always required.
In this case, the beam training overhead in \eqref{eq:overhead_training_general} is independent of the \gls{snr}, and we refer to this regime as the \gls{ior}.

Furthermore, it is worth noting that the \gls{dor} overhead regime is linked with the \gls{wbr} and the \gls{nbr} because \(\gls{SnrReceivedNormalizedAtLevel}\) follows the scaling laws in \eqref{eq:snr-scaling-law}.
In order to include these regimes in \eqref{eq:pilot_symbol_bound}, we adopt \eqref{eq:snr-normalized-wbr} and \eqref{eq:snr-normalized-nbr} in Appendix~\ref{sec:proof-snr-scaling} to rewrite \(\gls{SnrReceivedNormalizedAtLevel} = \gls{SnrReceivedFactorBoundWbr} \gls{NumRis} \gls{SizeCodebookAtLevel}\) and \(\gls{SnrReceivedNormalizedAtLevel} = \gls{SnrReceivedFactorBoundNbr} \glsSquared{NumRis}\) for the \gls{wbr} and the \gls{nbr}, respectively, where
\(
  \gls{SnrReceivedFactorBoundWbr} =
  \frac{\gls{PowerTransmitter} \gls{GainTransmitter} \gls{NumAp} \gls{GainReceiver}}{\gls{NoisePower}}
  \frac{\glsSquared{Wavelength} \gls{LengthUnitCellX} \gls{LengthUnitCellY} \gls{AreaRisFactor}}{(4 \pi \gls{DistanceIncident})^2 \gls{Area}}
\) and
\(
  \gls{SnrReceivedFactorBoundNbr} =
  \frac{\gls{PowerTransmitter} \gls{GainTransmitter} \gls{NumAp} \gls{GainReceiver}}{\gls{NoisePower}}
  \frac{\glsSquared{Wavelength} \glsSquared{Wavelength} \abs*{\gls{GainUc}}^2}{\left(4 \pi \gls{DistanceIncident}\right)^2 \left(4 \pi \gls{DistanceReflected}\right)^2}
\).
Note that the \(\max\) operator has been dropped because \eqref{eq:snr-normalized-wbr} and \eqref{eq:snr-normalized-nbr} hold for the codeword that optimally illuminates the user location, i.e., the codeword that maximizes \(\gls{SnrReceivedNormalized}\).
As a result, \eqref{eq:pilot_symbol_bound} can be rewritten as
\begin{equation}\label{eq:pilot_symbol_bound_regimes}
  \gls{NumSymTxLevel} \geq
  \left\{
    \begin{aligned}
      &\frac{\gls{SnrReceivedMin}}{\gls{SnrReceivedFactorBoundWbr} \gls{NumRis} \gls{SizeCodebookAtLevel}} & \qquad & \text{\gls{wbr}}\\
      &\frac{\gls{SnrReceivedMin}}{\gls{SnrReceivedFactorBoundNbr} \glsSquared{NumRis}} & \qquad & \text{\gls{nbr}}.
    \end{aligned}
  \right.
\end{equation}

\begin{remark}\label{rm:number_of_pilots_cancels_bandwidth}
  One can see from~\eqref{eq:overhead_training_general} that the beam training overhead can be reduced by reducing the pilot symbol duration \(\gls{DurationSymbolTraining}\).
  However, this is only possible for high \glspl{snr}, where \(\gls{NumSymTxLevel} \gg
  \gls{SnrReceivedMin}/\gls{SnrReceivedNormalizedAtLevel}\) holds, as can be seen from the following example.
  Assume that one pilot symbol is just sufficient for reliable beam training, i.e., we have \(\gls{SnrReceivedMin}/\gls{SnrReceivedNormalizedAtLevel} = 1\).
  In this case, we can use \(\gls{NumSymTxLevel} = 1\).
  On the other hand, \(\gls{SnrReceivedNormalizedAtLevel} \propto 1/\gls{NoisePower} \propto \gls{DurationSymbolTraining}\) because the noise power \(\gls{NoisePower}\) scales with the signal bandwidth and is thus inversely proportional to the symbol duration.
  Consequently, decreasing \(\gls{DurationSymbolTraining}\) leads to \(\gls{SnrReceivedMin}/\gls{SnrReceivedNormalizedAtLevel} > 1\) and \(\gls{NumSymTxLevel} > 1\) is required to satisfy~\eqref{eq:pilot_symbol_bound}.
  In other words, the positive effect of a shorter pilot symbol duration on the beam training overhead is compensated by the additional pilot symbols required.
\end{remark}

\section{Beam Training Analysis}\label{sec:beam-training-analysis}
\noindent
This section describes the considered \gls{fs}, \gls{hs}, and \gls{ts} beam training strategies, and provides detailed analysis of their respective overheads.
More specifically, we derive \(\gls{OverheadTraining}\) in \eqref{eq:overhead_training_general} and \(\gls{OverheadRelative}\) in \eqref{eq:overall-overhead} considering the different regimes for both overhead and achievable \gls{snr}.
In the following, we refer to a particular beam training strategy using superscript \(\mathsf{x}\in\{\text{\glsxtrshort{fs}},\text{\glsxtrshort{hs}},\text{\glsxtrshort{ts}}\}\).

\subsection{Beam Training Strategies}\label{sec:beam_training_strategies}
\noindent
For each beam training strategy, we first define the corresponding codeword sets \(\gls{SetCodewords}\) and \(\gls{SetCodewordsTraining}\). Then, using \eqref{eq:overhead_training_general} and \eqref{eq:pilot_symbol_bound_regimes}, we derive the overhead that is required for reliable beam training.

\subsubsection{Full Search}
\Gls{fs} beam training is the most robust scheme that selects the optimal codeword according to \eqref{eq:selected-codeword} because the entire codebook is searched and there is no dependence on previous decisions.
Hence, we have \(\gls{NumLevels} = 1\) and \(\gls{SetCodewords} = \gls{SetCodewordsTraining}\), which implies \(\gls{NumCodewordsTraining} = \gls{SizeCodebookAtLevel} = \gls{SizeCodebookAtLevelMax}\).
Then, the beam training overhead is given by
\begin{equation}\label{eq:overhead_training_full}
  \gls{OverheadTrainingFull}
  \geq \gls{TimeResponse} \left(1 + \gls{SizeCodebookAtLevelMax}\right)
  + \gls{DelayFeedback}
  + \gls{DurationSymbolTraining} \gls{OverheadFactorFull},
\end{equation}
where\footnote{Since the overhead only depends on the \gls{snr} if the system operates in the \gls{dor}, we consider three cases: (1) \gls{ior}, (2) \gls{dor} and \gls{wbr}, and (3) \gls{dor} and \gls{nbr}. Moreover, the overheads in the \gls{wbr} and the \gls{nbr} constitute upper and lower bounds for the transition regime, respectively, cf. Section~\ref{sec:achievable_snr}.}
\begin{equation}
  \gls{OverheadFactorFull} =
  \left\{
  \begin{aligned}
    &\gls{SizeCodebookAtLevelMax}
    & \; &\text{\gls{ior}},\\
    &\frac{1}{\gls{SnrReceivedFactorBoundWbr}}
    \frac{\gls{SnrReceivedMin}}{\gls{NumRis}}
    & \; &\text{\gls{dor}, \gls{wbr}},\\
    &\frac{\gls{SizeCodebookAtLevelMax}}{\gls{SnrReceivedFactorBoundNbr} \gls{NumRis}}
    \frac{\gls{SnrReceivedMin}}{\gls{NumRis}}
    & \; &\text{\gls{dor}, \gls{nbr}}.
  \end{aligned}
  \right.
\end{equation}

\subsubsection{Hierarchical Search}
\Gls{hs} beam training is a commonly adopted method that requires less overhead than \gls{fs} beam training~\cite{wei2022codebookdesignbeam,lv2023risaidedfield,wang2023hierarchicalcodebookbased}.
In this work, we adopt a hierarchically structured codebook with \(\gls{NumLevels} > 1\) levels and assume the following scheme.
Beam training starts with a full search over all codewords at level \(l=1\), i.e., \(\gls{SetCodewordsInitial} = \gls{SetCodewordsTrainingInitial}\) and \(\gls{NumCodewordsInitial} = \gls{NumCodewordsTrainingZero}\).
At each subsequent level, the training set \(\gls{SetCodewordsTraining}\) is determined by the subarea illuminated by codeword \(\gls{CodewordBestAtPreviousLevel}\), which is further divided into \(\gls{NumCodewordsFactor} = \abs*{\gls{SetCodewordsTraining}}\) parts.
Since \(\gls{NumCodewordsFactor}\) is usually small, we note that \(\gls{NumCodewordsInitial} > \gls{NumCodewordsFactor}\) may be required to achieve a sufficiently large \gls{snr} at the first codebook level.
Hence, the codebook size at the \(l\)th level is given by \(\abs*{\gls{SetCodewords}} = \gls{SizeCodebookAtLevel} = \gls{NumCodewordsInitial} \gls{NumCodewordsFactorPowerVar}\), but only \(\gls{NumCodewordsTotalHierarchical} = \gls{NumCodewordsInitial} + (\gls{NumLevels}-1)\gls{NumCodewordsFactor}\) codewords are used for beam training.
Consequently, the beam training overhead is given by
\begin{equation}\label{eq:overhead_training_hierarchical}
  \gls{OverheadTrainingHierarchical}
  \geq \gls{TimeResponse} \left(1 + \gls{NumCodewordsTotalHierarchical}\right)
  + \gls{DelayFeedback} \gls{NumLevels}
  + \gls{DurationSymbolTraining} \sum_{l=1}^{\gls{NumLevels}} \gls{OverheadFactorHierarchical},
\end{equation}
where
\begin{equation}
  \gls{OverheadFactorHierarchical} =
  \left\{
  \begin{aligned}
    &\gls{NumCodewordsTraining}
    & \; &\text{\gls{ior}},\\
    &\frac{\gls{NumCodewordsTraining}}{\gls{SnrReceivedFactorBoundWbr} \gls{SizeCodebookAtLevel}}
    \frac{\gls{SnrReceivedMin}}{\gls{NumRis}}
    & \; &\text{\gls{dor}, \gls{wbr}},\\
    &\frac{\gls{NumCodewordsTraining}}{\gls{SnrReceivedFactorBoundNbr} \gls{NumRis}}
    \frac{\gls{SnrReceivedMin}}{\gls{NumRis}}
    & \; &\text{\gls{dor}, \gls{nbr}}.
  \end{aligned}
  \right.
\end{equation}

\subsubsection{Tracking-Based Search}
In order include a beam training strategy with very low overhead in our analysis, we adopt a tracking-based scheme that exploits prior information, e.g., an approximate user location.
Similar to \gls{fs} beam training, \gls{ts} beam training uses a single-level codebook, i.e., \(\gls{NumLevels} = 1\) and \(\gls{SizeCodebookAtLevel} = \gls{SizeCodebookAtLevelMax}\).
However, \gls{ts} beam training exploits knowledge of the codeword that has been selected in the previous frame, in particular, \(\gls{SetCodewordsTraining}\) is determined by the local environment of the previously illuminated subarea~\cite{laue2022performancetradeoffris}.
Assuming that the subareas are arranged in a regular grid, the local environment is defined by \(\gls{NumCodewordsTraining} = 8\) codewords.
Ideally, scanning the local environment is sufficient to continuously track the mobile user without reinitialization\footnote{Reinitialization and initial beam training in general require beam training strategies such as \gls{fs} or \gls{hs}, which cause additional overhead.}.
Thus, the beam training overhead is given by
\begin{equation}\label{eq:overhead_training_tracking}
  \gls{OverheadTrainingTracking}
  \geq \gls{TimeResponse} \left(1 + 8\right)
  + \gls{DelayFeedback}
  + \gls{DurationSymbolTraining} \gls{OverheadFactorTracking},
\end{equation}
where
\begin{equation}
  \gls{OverheadFactorTracking} =
  \left\{
  \begin{aligned}
    &8
    & \; &\text{\gls{ior}},\\
    &\frac{
      8
    }{
      \gls{SnrReceivedFactorBoundWbr}\gls{SizeCodebookAtLevelMax}
    }
    \frac{\gls{SnrReceivedMin}}{\gls{NumRis}}
    & \; &\text{\gls{dor}, \gls{wbr}},\\
    &\frac{
      8
    }{
      \gls{SnrReceivedFactorBoundNbr}\gls{NumRis}
    }
    \frac{\gls{SnrReceivedMin}}{\gls{NumRis}}
    & \; &\text{\gls{dor}, \gls{nbr}}.
  \end{aligned}
  \right.
\end{equation}

\subsection{Overall Overhead}
\noindent
Next, we use the results in \eqref{eq:overhead_training_full}, \eqref{eq:overhead_training_hierarchical}, and \eqref{eq:overhead_training_tracking} to define the overall overhead \(\gls{OverheadRelative}\) for \gls{fs}, \gls{hs}, and \gls{ts} beam training.
In contrast to beam training overhead \(\gls{OverheadTraining}\), overall overhead \(\gls{OverheadRelative}\) includes the frame duration and the channel estimation overhead, which are relevant for the effective rate in \eqref{eq:rate_effective}.

Recall that frame duration \(\gls{DurationFrame}\) corresponds to the duration for which the mobile user is assumed to stay within one subarea, which is mainly determined by user velocity \(\gls{Velocity}\), the user trajectory, and the sizes of the subareas.
Since the user trajectory is usually unknown~\cite{zhang2022rateoverheadtradeoff}, we propose a parameterized model for the frame duration as follows.
As a reference, let \(\gls{PathThroughSubareaMin}=\min_{m \in \gls{SetCodewordsFinal}} \gls{PathThroughSubarea}\) denote the shortest diagonal path through a subarea.
Then, we assume that the user travels a distance of \(p = \gls{FactorFrameDesign} \gls{PathThroughSubareaMin}\) through a subarea, where frame factor \(\gls{FactorFrameDesign} > 0\) is a design parameter\footnote{Although the user's path through a subarea may be larger than \gls{PathThroughSubareaMin}, \(\gls{FactorFrameDesign} < 1\) results in a good balance between beam training overhead and correct beam alignment. A similar approach was adopted in \cite{yang2018hierarchicalcodebookbeam}, where \(\gls{FactorFrameDesign}\) was used to reduce the beam coherence time in order to capture impairments in practical systems.} to account for arbitrary user trajectories.
Moreover, we note that \(\gls{PathThroughSubareaMin} \propto 1/\gls{SizeCodebookAtLevelMax}\) if a given coverage area is fully illuminated by \(\gls{SizeCodebookAtLevelMax}\) beams.
Thus, we define \(\gls{PathThroughSubareaMin} = \gls{RatioAreaCodebook}/\gls{SizeCodebookAtLevelMax}\), where \(\gls{RatioAreaCodebook}\) is specified by the considered sizes and shapes of the subareas.
For example, as assumed for Lemma~\ref{lm:snr-regimes}, if the coverage area is partitioned into rectangular subareas of equal size, \(\gls{PathThroughSubareaMin} = \sqrt{(\gls{AreaY}/\gls{SizeCodebookAtLevelMaxY})^2 + (\gls{AreaZ}/\gls{SizeCodebookAtLevelMaxX})^2}\) and \(\gls{RatioAreaCodebook} = \sqrt{\left(\gls{AreaY}\gls{SizeCodebookAtLevelMaxX}\right)^2+\left(\gls{AreaZ}\gls{SizeCodebookAtLevelMaxY}\right)^2}\), cf. \eqref{eq:diameter}.
Based on these assumptions, the frame duration is parameterized as follows
\begin{equation}\label{eq:duration-frame}
  \gls{DurationFrame}
  = \frac{p}{\gls{Velocity}}
  = \frac{\gls{FactorFrameDesign}}{\gls{Velocity}} \frac{\gls{RatioAreaCodebook}}{\gls{SizeCodebookAtLevelMax}}.
\end{equation}
Substituting~\eqref{eq:duration-frame} in~\eqref{eq:overall-overhead} results in
\(
  \gls{OverheadRelative}
  = \gls{OverheadTraining} \frac{\gls{Velocity} \gls{SizeCodebookAtLevelMax}}{\gls{FactorFrameDesign} \gls{RatioAreaCodebook}} \left(1 - \gls{OverheadEstimationNormalized}\right) + \gls{OverheadEstimationNormalized}
\),
which is rewritten for notational simplicity as
\begin{equation}\label{eq:bound_overhead_training}
  \gls{OverheadRelative} = \gls{SimplificationOverhead} \gls{SizeCodebookAtLevelMax}
  \gls{OverheadTraining}
   + \gls{OverheadEstimationNormalized},
\end{equation}
where \(\gls{SimplificationOverhead} = \frac{\gls{Velocity}}{\gls{FactorFrameDesign} \gls{RatioAreaCodebook}}\left(1 - \gls{OverheadEstimationNormalized}\right)\).
Substituting \eqref{eq:overhead_training_full}, \eqref{eq:overhead_training_hierarchical}, and \eqref{eq:overhead_training_tracking} in \eqref{eq:bound_overhead_training} results in the following overall overheads for reliable beam training
\begin{align}
  \gls{OverheadRelativeFull}
  &\geq \gls{SimplificationOverhead} \gls{SizeCodebookAtLevelMax} \left[
    \gls{TimeResponse} \left(1 + \gls{SizeCodebookAtLevelMax}\right)
    + \gls{DelayFeedback}
    + \gls{DurationSymbolTraining} \gls{OverheadFactorFull}
  \right] + \gls{OverheadEstimationNormalized},\label{eq:overhead_relative_full}\\
  \gls{OverheadRelativeHierarchical}
  &\geq \gls{SimplificationOverhead} \gls{SizeCodebookAtLevelMax} \left[
    \gls{TimeResponse} \left(1 + \gls{NumCodewordsTotalHierarchical}\right)
    + \gls{DelayFeedback} \gls{NumLevels}
    + \gls{DurationSymbolTraining} \sum_{l=1}^{\gls{NumLevels}} \gls{OverheadFactorHierarchical}
  \right] + \gls{OverheadEstimationNormalized},\label{eq:overhead_relative_hierarchical}\\
  \gls{OverheadRelativeTracking}
  &\geq \gls{SimplificationOverhead} \gls{SizeCodebookAtLevelMax} \left[
    \gls{TimeResponse} \left(1 + 8\right)
    + \gls{DelayFeedback}
    + \gls{DurationSymbolTraining} \gls{OverheadFactorTracking}
  \right] + \gls{OverheadEstimationNormalized}.\label{eq:overhead_relative_tracking}
\end{align}

\section{Performance Tradeoff}\label{sec:performance-tradeoff}
\noindent
In this section, we study the performance tradeoff between the overall overhead and the achievable \gls{snr} for \gls{ris} beam training.
First, we analyze the relevance of the overhead and provide an upper bound on \(\gls{Velocity}\), which guarantees that codebook-based \gls{ris} configuration has a negligible impact on the system performance.
Then, we show how \gls{fs}, \gls{hs}, and \gls{ts} beam training lead to different effective rates if the overhead is not negligible.

\subsection{Relevance of Overhead}
\noindent
From \eqref{eq:rate_effective}, one can see that a performance tradeoff between overhead and achievable \gls{snr} only exists if the overall overhead \(\gls{OverheadRelative}\) is sufficiently large.
Moreover, we note that the overall overhead is strongly influenced by the mobile user's velocity because it affects both the subframe duration and the frame duration.
This fundamental relation between \(\gls{OverheadRelative}\) and \(\gls{Velocity}\) is described by the following proposition.

\begin{proposition}\label{prop:overhead_form}
  The overall overhead is a quadratic function of the user velocity and has the form \(\gls{OverheadRelative} = \left(\gls{ParabolaParameterOne} + \gls{ParabolaParameterTwo}\right) \gls{Velocity} - \gls{ParabolaParameterOne} \gls{ParabolaParameterTwo} \glsSquared{Velocity}\), where \(\gls{ParabolaParameterOne} = \frac{\gls{SizeCodebookAtLevelMax} \gls{OverheadTraining}}{\gls{FactorFrameDesign} \gls{RatioAreaCodebook}}\) and \(\gls{ParabolaParameterTwo} = \frac{\gls{OverheadEstimation} 4 \sqrt{\pi}}{3 \gls{Wavelength}}\).
  \(\gls{OverheadRelative}\) is a concave parabola and the maximum value is greater or equal to one.
\end{proposition}

\begin{proof}
  By definition of \(\gls{SimplificationOverhead}\) and \(\gls{OverheadEstimationNormalized}\), \eqref{eq:bound_overhead_training} can be written as
  \(
    \gls{OverheadRelative}
    = \frac{\gls{Velocity}}{\gls{FactorFrameDesign} \gls{RatioAreaCodebook}}\left(1 - \frac{\gls{OverheadEstimation} 4 \sqrt{\pi} \gls{Velocity}}{3 \gls{Wavelength}}\right) \gls{SizeCodebookAtLevelMax} \gls{OverheadTraining}
    + \frac{\gls{OverheadEstimation} 4 \sqrt{\pi} \gls{Velocity}}{3 \gls{Wavelength}}
    = \left(\gls{ParabolaParameterOne} + \gls{ParabolaParameterTwo}\right) \gls{Velocity}
      - \gls{ParabolaParameterOne} \gls{ParabolaParameterTwo} \glsSquared{Velocity}
  \) for \(\gls{ParabolaParameterOne} = \frac{\gls{SizeCodebookAtLevelMax} \gls{OverheadTraining}}{\gls{FactorFrameDesign} \gls{RatioAreaCodebook}}\) and \(\gls{ParabolaParameterTwo} = \frac{\gls{OverheadEstimation} 4 \sqrt{\pi}}{3 \gls{Wavelength}}\).
  Since \(\gls{ParabolaParameterOne} \gls{ParabolaParameterTwo} > 0\), \(\gls{OverheadRelative}\) is a concave parabola (opening to the bottom).
  Thus, the maximum value is given by the vertex of the parabola, which is found by evaluating \(\frac{d}{d \gls{Velocity}} \gls{OverheadRelative} = 0\) and given by \(\max_{\gls{Velocity}} \gls{OverheadRelative} = \frac{(a+b)^2}{4ab}\) at \(\gls{Velocity} = \frac{a+b}{2ab}\).
  Since \(\gls{ParabolaParameterOne} \gls{ParabolaParameterTwo} > 0\), one can rewrite \(\frac{(a+b)^2}{4ab} \geq 1\) as \((a-b)^2 \geq 0\), which always holds.
  Therefore, \(\max_{\gls{Velocity}} \gls{OverheadRelative} \geq 1\).
\end{proof}

We note that, although \(\gls{OverheadRelative}\) is a quadratic function of the user velocity, the overall overhead does not quadratically grow with \(\gls{Velocity}\).
Instead, since \(\frac{d}{d \gls{Velocity}} \gls{OverheadRelative} = \gls{ParabolaParameterOne} + \gls{ParabolaParameterTwo} - 2 \gls{ParabolaParameterOne} \gls{ParabolaParameterTwo} \gls{Velocity}\), the gradient of \(\gls{OverheadRelative}\) is a linearly decreasing function with a maximum value of \(\gls{ParabolaParameterOne} + \gls{ParabolaParameterTwo}\) at \(\gls{Velocity} = 0\).

Furthermore, we note that the impact of \(\gls{OverheadRelative}\) on the effective rate in \eqref{eq:rate_effective} is negligible if \(\gls{OverheadRelative} < \gls{OverheadThreshold}\), \(0 \leq \gls{OverheadThreshold} \ll 1\).
Based on Proposition~\ref{prop:overhead_form}, this condition for the overall overhead leads to an upper bound for the mobile user's velocity.

\begin{lemma}\label{lm:negligible_overhead}
  The overall overhead has negligible impact, i.e., \(\gls{OverheadRelative} < \gls{OverheadThreshold}\), on the effective rate if
  \(
    \gls{Velocity}
    < \frac{
      \gls{ParabolaParameterOne} + \gls{ParabolaParameterTwo}
      - \sqrt{(\gls{ParabolaParameterOne}+\gls{ParabolaParameterTwo})^2 - 4\gls{ParabolaParameterOne}\gls{ParabolaParameterTwo}\gls{OverheadThreshold}}
    }{
      2\gls{ParabolaParameterOne}\gls{ParabolaParameterTwo}
  }\).
\end{lemma}

\begin{proof}
  Since \(\gls{OverheadRelative}\) is a quadratic function in \(\gls{Velocity}\), solving \(\gls{OverheadRelative} < \gls{OverheadThreshold}\) for \(\gls{Velocity}\) leads to the lower bound \(\gls{Velocity} > \frac{\gls{ParabolaParameterOne} + \gls{ParabolaParameterTwo} + \sqrt{(\gls{ParabolaParameterOne}+\gls{ParabolaParameterTwo})^2 - 4\gls{ParabolaParameterOne}\gls{ParabolaParameterTwo}\gls{OverheadThreshold}}}{2\gls{ParabolaParameterOne}\gls{ParabolaParameterTwo}}\) and the upper bound \(\gls{Velocity} < \frac{\gls{ParabolaParameterOne} + \gls{ParabolaParameterTwo} - \sqrt{(\gls{ParabolaParameterOne}+\gls{ParabolaParameterTwo})^2 - 4\gls{ParabolaParameterOne}\gls{ParabolaParameterTwo}\gls{OverheadThreshold}}}{2\gls{ParabolaParameterOne}\gls{ParabolaParameterTwo}}\), which are always distinct because \((\gls{ParabolaParameterOne}+\gls{ParabolaParameterTwo})^2 - 4\gls{ParabolaParameterOne}\gls{ParabolaParameterTwo}\gls{OverheadThreshold} > 0\) due to \(\gls{OverheadThreshold} < 1\) and \(\frac{(\gls{ParabolaParameterOne}+\gls{ParabolaParameterTwo})^2}{4\gls{ParabolaParameterOne}\gls{ParabolaParameterTwo}} \geq 1\).
  However, only the upper bound for \(\gls{Velocity}\) is a feasible solution because the lower bound leads to zero time for data transmission, see Appendix~\ref{sec:proof-negligible-overhead}.
\end{proof}

Lemma~\ref{lm:negligible_overhead} can be applied to \gls{fs}, \gls{hs}, or \gls{ts} beam training, each resulting in a different upper bound on \(\gls{Velocity}\).
In order to remove this dependence on a particular beam training strategy, one can find a generalized bound on \(\gls{Velocity}\) as follows.

\begin{corollary}\label{cr:condition_no_impact}
  Regardless of whether \gls{fs}, \gls{hs}, or \gls{ts} beam training is employed, the overhead for codebook-based \gls{ris} configuration has negligible impact on the effective rate if
  \begin{equation}\label{eq:velocity_bound}
    \gls{Velocity}
    < \gls{VelocityBound}
    = \frac{
      \gls{ParabolaParameterOneBound} + \gls{ParabolaParameterTwo}
      - \sqrt{
        (\gls{ParabolaParameterOneBound}+\gls{ParabolaParameterTwo})^2
        - 4\gls{ParabolaParameterOneBound}\gls{ParabolaParameterTwo}\gls{OverheadThreshold}
      }
    }{
      2\gls{ParabolaParameterOneBound}\gls{ParabolaParameterTwo}
    },
  \end{equation}
  where \(\gls{ParabolaParameterOneBound} = \frac{\gls{SizeCodebookAtLevelMax} \gls{OverheadTrainingBound}}{\gls{FactorFrameDesign} \gls{RatioAreaCodebook}}\) and \(\gls{OverheadTrainingBound} = 1.25 \gls{NumSymTxLevelMax} \gls{SizeCodebookAtLevelMax} \gls{DurationSymbolTraining}
  + \left(1 + \gls{SizeCodebookAtLevelMax}\right) \gls{TimeResponse}
  + \gls{NumLevels} \gls{DelayFeedback}\).
\end{corollary}

\begin{proof}
  Based on \eqref{eq:bound_overhead_training}, one can see that \(\gls{OverheadRelative}\) increases with \(\gls{OverheadTraining}\) because \(\gls{SimplificationOverhead} > 0\).
  Thus, the overall overhead is bounded as
  \(
    \gls{OverheadRelative} \leq \gls{OverheadRelativeBound}
    = \gls{SimplificationOverhead} \gls{SizeCodebookAtLevelMax} \gls{OverheadTrainingBound}
    + \gls{OverheadEstimationNormalized}
  \), where
  \(
    \gls{OverheadTrainingBound}
    \geq \max_{\mathsf{x} \in \{\text{\gls{fs}},\text{\gls{hs}},\text{\gls{ts}}\}} \gls{OverheadTrainingX}
  \) denotes an upper bound for the beam  training overhead.
  For the problem at hand, we find \(\gls{OverheadTrainingBound} = 1.25 \gls{NumSymTxLevelMax} \gls{SizeCodebookAtLevelMax} \gls{DurationSymbolTraining}
  + \left(1 + \gls{SizeCodebookAtLevelMax}\right) \gls{TimeResponse}
  + \gls{NumLevels} \gls{DelayFeedback}\), see Appendix~\ref{sec:proof-overhead-bound}.
  Therefore, \(\gls{Velocity} < \frac{\gls{ParabolaParameterOneBound} + \gls{ParabolaParameterTwo} - \sqrt{(\gls{ParabolaParameterOneBound}+\gls{ParabolaParameterTwo})^2 - 4\gls{ParabolaParameterOneBound}\gls{ParabolaParameterTwo}\gls{OverheadThreshold}}}{2\gls{ParabolaParameterOneBound}\gls{ParabolaParameterTwo}}\) results from Lemma~\ref{lm:negligible_overhead}, where \(\gls{OverheadRelative}\) is replaced by \(\gls{OverheadRelativeBound}\).
\end{proof}

To illustrate the results of Corollary~\ref{cr:condition_no_impact}, we show \(\gls{OverheadRelativeBound}\) as a function of \(\gls{Velocity}\) in Fig.~\ref{fig:relevance_analytical}, where different codebook sizes and \gls{ris} response times are considered.
As one can see, the overall overhead does not only increase with the codebook size, but also significantly depends on the \gls{ris} response time \(\gls{TimeResponse}\).
For example, for a fast \gls{ris} with \(\gls{TimeResponse} = \qty{1}{\micro\second}\), Fig.~\ref{fig:relevance_analytical} shows that the maximum overhead is below \qty{10}{\percent} if \(\gls{Velocity} < \qty{20}{\kilo\metre/\hour}\) and \(\gls{SizeCodebookAtLevelMax} \leq 512\).
In contrast, for a slow \gls{ris} with \(\gls{TimeResponse} = \qty{1}{\milli\second}\), the same low overhead is only achieved if \(\gls{Velocity} < \qty{0.45}{\kilo\metre/\hour}\).
However, we note that slow \glspl{ris} with large \(\gls{TimeResponse}\) can still be efficiently configured using beam training if \(\gls{OverheadRelative} \ll \gls{OverheadRelativeBound}\).
For example, our numerical results in Section~\ref{sec:numerical_results} show that \gls{ts} beam training achieves high effective rates for \(\gls{Velocity} > \qty{20}{\kilo\metre/\hour}\) and \(\gls{TimeResponse} = \qty{1}{\milli\second}\), see Fig.~\ref{fig:relevance}.

\begin{figure}[!t]
  \centering
  \includegraphics{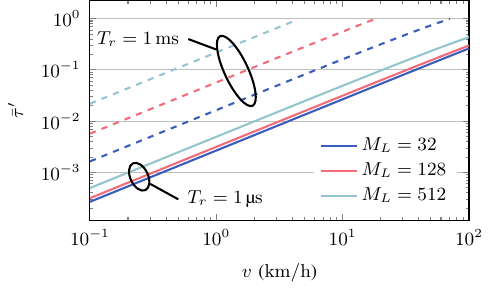}
  \caption{Upper bound for overall overhead as a function of the user velocity, considering different codebook sizes and \gls{ris} response times \(\gls{TimeResponse}\) (cf. Table~\ref{tab:technology}); frame factor \(\gls{FactorFrameDesign} = 0.15\), \(\gls{Area} = \qtyproduct{100x50}{\metre}\), \(\gls{DurationSymbolTraining} = \qty{1}{\micro\second}\), \(\gls{DelayFeedback} = \qty{0.1}{\milli\second}\), and \(\gls{NumSymTxLevelMax} = 8\).}\label{fig:relevance_analytical}
\end{figure}

\subsection{Tradeoff Analysis}
\noindent
In the previous section, we have shown that for certain system parameters the performance may not be affected by the beam training overhead.
In this section, we reveal how non-negligible overheads lead to different effective rates for \gls{fs}, \gls{hs}, and \gls{ts}.
Thereby, we assume \(\gls{NumCodewordsInitial} > 8\), which is a necessary condition for the considered \gls{ts} beam training and may be required for providing sufficiently large reflection gains in practice~\cite{wang2022beamtrainingalignment}.
For \gls{hs} beam training, we assume \(\gls{NumCodewordsFactor} > 1\) and \(\gls{NumLevels} > 1\), and note that the codebook at level \(\gls{NumLevels}\) is adopted for the single-level schemes, i.e., \gls{fs} and \gls{ts} beam training.

\begin{lemma}\label{lm:ts_superior}
  Assuming optimal codeword selection, \gls{ts} beam training leads to higher effective rates than \gls{fs} and \gls{hs} beam training.
  Moreover, \gls{hs} beam training outperforms \gls{fs} beam training or vice versa, depending on the considered system parameters.
\end{lemma}

\begin{proof}
  Due to optimal codeword selection and identical codebooks at the highest codebook level, all considered beam training strategies choose the same codeword and provide the same \gls{snr} for data transmission.
  Thus, one can see from \eqref{eq:rate_effective} that \gls{fs}, \gls{hs}, and \gls{ts} only result in different effective rates due to different overheads.
  Consequently, \gls{ts} provides the highest effective rate because \(\gls{OverheadRelativeTracking} < \gls{OverheadRelativeFull}\) and \(\gls{OverheadRelativeTracking} < \gls{OverheadRelativeHierarchical}\), cf. \cref{eq:overhead_relative_full,eq:overhead_relative_hierarchical,eq:overhead_relative_tracking}.
  However, \(\gls{OverheadRelativeFull} > \gls{OverheadRelativeHierarchical}\) or \(\gls{OverheadRelativeFull} < \gls{OverheadRelativeHierarchical}\) may hold, depending on the system parameters in \eqref{eq:overhead_relative_full} and \eqref{eq:overhead_relative_hierarchical}.
\end{proof}

Interestingly, Lemma~\ref{lm:ts_superior} reveals that \gls{fs} beam training outperforms \gls{hs} beam training if \(\gls{OverheadRelativeFull} < \gls{OverheadRelativeHierarchical}\), which particularly happens, e.g., for fast \gls{ris} response times and large feedback delays.
Further insight on this case is provided by the following lemma, where the above inequality is simplified to a lower bound on \(\gls{DelayFeedback}\).

\begin{lemma}\label{lm:fs_better_hs}
  \Gls{fs} beam training provides a higher effective rate than \gls{hs} beam training if
  \begin{equation}\label{eq:condition_fs_better_hs}
    \gls{DelayFeedback} > \gls{DelayFeedbackBound}
    = \frac{\gls{SizeCodebookAtLevelMax} - \gls{NumCodewordsTotalHierarchical}}{\gls{NumLevels} - 1}
    \left(\gls{TimeResponse} + \gls{DurationSymbolTraining} \gls{NumSymTxLevelMax}\right).
  \end{equation}
\end{lemma}

\begin{proof}
  Recall that the codebook for \gls{fs} beam training is identical to that for \gls{hs} beam training at level \(\gls{NumLevels}\).
  Thus, based on \eqref{eq:overhead_training_general} and \eqref{eq:overall-overhead}, \(\gls{OverheadRelativeFull} < \gls{OverheadRelativeHierarchical}\) is equivalent to
  \begin{equation}\label{eq:condition_fs_better_hs_proof}
    \gls{DelayFeedback} \left(\gls{NumLevels} - 1\right)
    > \gls{TimeResponse} \left(\gls{SizeCodebookAtLevelMax} - \gls{NumCodewordsTotalHierarchical}\right)
    + \gls{DurationSymbolTraining} \left(
      \gls{NumSymTxLevelMax} \gls{SizeCodebookAtLevelMax}
      - \sum_{l=1}^{\gls{NumLevels}} \gls{NumSymTxLevel} \gls{NumCodewordsTraining}
    \right).
  \end{equation}
  Moreover, \eqref{eq:snr-scaling-law} and \eqref{eq:pilot_symbol_bound} imply \(\gls{NumSymTxLevel} \geq \gls{NumSymTxLevelNext}\), which can be used to reformulate \eqref{eq:condition_fs_better_hs_proof} as follows
  \begin{equation}\label{eq:condition_fs_better_hs_proof_relaxed}
    \gls{DelayFeedback} \left(\gls{NumLevels} - 1\right)
    > \gls{TimeResponse} \left(\gls{SizeCodebookAtLevelMax} - \gls{NumCodewordsTotalHierarchical}\right)
    + \gls{DurationSymbolTraining} \left(
      \gls{NumSymTxLevelMax} \gls{SizeCodebookAtLevelMax}
      - \gls{NumSymTxLevelMax} \sum_{l=1}^{\gls{NumLevels}} \gls{NumCodewordsTraining}
    \right).
  \end{equation}
  Finally, \eqref{eq:condition_fs_better_hs_proof_relaxed} is simplified to \eqref{eq:condition_fs_better_hs} by using \(\gls{NumCodewordsTotalHierarchical} = \sum_{l=1}^{\gls{NumLevels}} \gls{NumCodewordsTraining}\).
\end{proof}

It is worth noting that the above results are based on the assumption of reliable beam training, which cannot always be guaranteed in practice, cf. Remark~\ref{rm:impact_on_decisions}.
However, since reliable beam training is desired, practical systems are designed such that the number of erroneous codeword selections is small.
Thus, it is expected that Lemmas~\ref{lm:ts_superior} and~\ref{lm:fs_better_hs} generally hold, which is verified by our numerical results in Section~\ref{sec:numerical_results}.

\section{Numerical Results}\label{sec:numerical_results}
\noindent
In the previous sections, our analysis of the tradeoff between the achievable \gls{snr} and the beam training overhead was based on ideal illumination of the subareas and reliable codeword selection during beam training.
In this section, these results are verified by numerical simulations, considering a Rician fading channel model, non-ideal phase-shift designs, and random user trajectories.
For the following evaluation, we employ numerical simulations to compute the \gls{snr} for data transmission in \eqref{eq:snr-data} and the effective ergodic rate in \eqref{eq:rate_effective}, averaged over the channel realizations for a random user trajectory with a total length of \(\qty{1000}{\metre}\).
As a benchmark and upper bound for the beam training performance, the \gls{ris} is configured with narrow beams that always focus on the known user location assuming zero training overhead.

The user trajectory follows the random direction mobility model~\cite{camp2002surveymobilitymodels}, where the user travels along straight paths within the coverage area and randomly changes its direction at the boundary of the coverage area.

Instead of the idealized phase-shift design that perfectly reflects the power collected at the \gls{ris} to the targeted subareas, our numerical results are based on a practical phase-shift design to realize narrow and wide beams, respectively.
More specifically, the codewords of the \gls{ris} beam codebook are realized by the quadratic phase-shift profile from~\cite{jamali2021powerefficiencyoverhead}, which is parameterized to illuminate a particular subarea of the coverage area.
Although the achieved coverage can be improved with optimized phase-shift designs, we adopt the quadratic phase-shift design due to its analytical form and note that the performance gap to optimized designs is generally small~\cite{laue2021irsassistedactive,ghanem2022optimizationbasedphasea}.

Furthermore, we consider a geometric Rician fading channel model, where the \gls{bs}-\gls{ris} channel and the \gls{ris}-user channel are given by
\begin{align}
  \mt{H}_i &= \sqrt{\gls{PathlossIncident}} (\bar{\mt{H}}_i + \tilde{\mt{H}}_i)\label{eq:channel_incident}\\
  \vt{h}_r &= \sqrt{\gls{PathlossReflected}} (\bar{\vt{h}}_r + \tilde{\vt{h}}_r)\label{eq:channel_reflected},
\end{align}
respectively, and
\begin{align}
  \bar{\mt{H}}_i &= \sqrt{\gls{FactorRicianIncident}/(\gls{FactorRicianIncident} + 1)} e^{j\gls{PhaseLosIncident}}
  \gls{ArrayResponseArrivalRis}(\gls{AngleIncidentElevation}, \gls{AngleIncidentAzimuth})
  \glsConjTrans{ArrayResponseDepartureAp}(\omega)\label{eq:channel_incident_los}\\
  \tilde{\mt{H}}_i &= \sqrt{1/((\gls{FactorRicianIncident} + 1) \gls{NumScatterersInc})}
  \sum_{c=1}^{\gls{NumScatterersInc}} \gls{ChannelNlosIncident}
    \gls{ArrayResponseArrivalRis}(\theta_{i,c}, \phi_{i,c})
    \glsConjTrans{ArrayResponseDepartureAp}(\omega_c)\label{eq:channel_incident_nlos}\\
  \bar{\vt{h}}_r &= \sqrt{\gls{FactorRicianReflected}/(\gls{FactorRicianReflected} + 1)} e^{-j\gls{PhaseLosReflected}} \gls{ArrayResponseDepartureRis}(\gls{AngleReflectedElevation}, \phi_r)\label{eq:channel_reflected_los}\\
  \tilde{\vt{h}}_r &= \sqrt{1/((\gls{FactorRicianReflected} + 1) \gls{NumScatterersRef})} \sum_{c=1}^{\gls{NumScatterersRef}} \gls{ChannelNlosReflected} \gls{ArrayResponseDepartureRis}(\theta_{r,c}, \phi_{r,c}).\label{eq:channel_reflected_nlos}
\end{align}
In \cref{eq:channel_incident,eq:channel_incident_los,eq:channel_incident_nlos,eq:channel_reflected,eq:channel_reflected_los,eq:channel_reflected_nlos}, for \(s\in\{i,r\}\), \(\gls{Pathloss}\), \(\gls{FactorRician}\), \(\gls{PhaseLos}\), and \(\gls{NumScatterers}\) denote the path loss, Rician K-factor, phase of the \gls{los} path, and number of \gls{nlos} paths, respectively, and \(\gls{AodBs}\) denotes the \gls{aod} at the \gls{bs}.
Moreover, the \glspl{aoa} and \glspl{aod} of the \gls{nlos} paths are indexed by \(c \in \{1, \dotsc, \gls{NumScatterers}\}\), and the path coefficient of the \(c\)th \gls{nlos} path is given by \(\gls{ChannelNlos} \sim \CN{0, 1}\).
The \((q_x, q_y)\)th element of \(\gls{ArrayResponseArrivalRis}(\gls{AngleIncidentElevation}, \gls{AngleIncidentAzimuth})\) and \(\gls{ArrayResponseDepartureRis}(\theta_r, \phi_r)\) as well as the \(n\)th element of \(\gls{ArrayResponseDepartureAp}(\gls{AodBs})\) are given by
\(
  e^{j\frac{2\pi}{\gls{Wavelength}}(d_x q_x \sin(\gls{AngleIncidentElevation}) \cos(\phi_i) + d_y q_y \sin(\gls{AngleIncidentElevation}) \sin(\phi_i))}
\),
\(
  e^{-j\frac{2\pi}{\gls{Wavelength}}(d_x q_x \sin(\theta_r) \cos(\phi_r) + d_y q_y \sin(\theta_r) \sin(\phi_r))}
\)
and
\(e^{-j\frac{2\pi}{\gls{Wavelength}} d n \sin(\omega)}\), respectively, where \(q_x, q_y\in\{0, \dotsc, \sqrt{\gls{NumRis}}-1\}\) denote two-dimensional indices of the \gls{ris} unit cells and \(n\in\{0, \dotsc, \gls{NumAp}-1\}\).

\subsection{System Parameters}
\noindent
Our theoretical analysis has shown that the performance tradeoff between overhead and achievable \gls{snr} depends on several different system parameters, including the codebook size, coverage area size, \gls{ris} size, transmit power, and number of pilot symbols.
Therefore, in the following, we focus on typical \gls{ris} deployments and realistic timing constraints in order to evaluate practical scenarios.
We consider user velocities \(\gls{Velocity} \in [\qty{3}{\kilo\metre/\hour},\qty{100}{\kilo\metre/\hour}]\) and \gls{ris} response times \(\gls{TimeResponse} \in [\qty{100}{\nano\second},\qty{1}{\milli\second}]\) in order to capture different mobility scenarios and \gls{ris} technologies.
Unless otherwise noted, the feedback delay is set to \(\gls{DelayFeedback} = \qty{0.1}{\milli\second}\), which is targeted by 3GPP for future systems~\cite{hassan2021keytechnologiesultra,ji2021severalkeytechnologies}.

We adopt codebook sizes up to \(\gls{SizeCodebookAtLevel} = 2048\), which corresponds to a hierarchical codebook with up to \(\gls{NumLevels} = 5\) levels assuming \(\gls{NumCodewordsInitial} = 8\) and \(\gls{NumCodewordsFactor} = 4\).
When comparing the performances of the considered beam training strategies, \gls{fs} and \gls{ts} beam training use the same codewords as specified by the highest codebook level for \gls{hs} beam training.

Moreover, the \gls{ris}, the \gls{bs}, and the center of the coverage area are located at coordinates \((\qty{0}{\metre}, \qty{0}{\metre}, \qty{0}{\metre})\), \((\qty{0}{\metre}, \qty{40}{\metre}, \qty{50}{\metre})\), and \((\qty{-10}{\metre}, \qty{-20}{\metre}, \qty{100}{\metre})\), respectively.
The coverage area has dimensions \(\gls{AreaY} = \qty{100}{\metre}\) and \(\gls{AreaZ} = \qty{50}{\metre}\) and is partitioned into subareas of equal size, resulting in \(\gls{RatioAreaCodebook} = \sqrt{\left(\gls{AreaY}\gls{SizeCodebookAtLevelMaxX}\right)^2+\left(\gls{AreaZ}\gls{SizeCodebookAtLevelMaxY}\right)^2}\).
The \gls{ris} unit cell and \gls{bs} antenna spacing are given by \(d=d_x=d_y=\gls{Wavelength} /2\).
The \gls{bs} has \(\gls{NumAp} = 16\) antenna elements and steers its main lobe towards the center of the \gls{ris}.
For the channels, we assume \(\gls{NumScatterersInc}=\gls{NumScatterersRef}=6\) and \(K_i = K_r = 4\), and \(\gls{PathlossIncident}\) and \(\gls{PathlossReflected}\) are given by the free-space path loss model.
The azimuth angles \((\phi_{i,c},\phi_{r,c})\) and elevation angles \((\theta_{i,c},\theta_{r,c})\) of the \gls{nlos} paths are drawn from uniform distributions \(\mathcal{U}(0,2\pi)\) and \(\mathcal{U}(0,\pi/2)\), respectively.
Furthermore, we assume antenna element gains \(\gls{GainTransmitter} = \gls{GainReceiver} = 1\), transmit power \(\gls{PowerTransmitter} = \qty[qualifier-mode=combine]{15}{\deci\bel\of{m}}\), channel estimation overhead \(\gls{OverheadEstimation} = 40\gls{DurationSymbolTraining}\)~\cite{wang2020compressedchannelestimation}, frame factor \(\gls{FactorFrameDesign} = 0.15\), carrier frequency \SI{28}{\giga\hertz}, signal bandwidth \(\gls{BandwidthPilot} = 1/\gls{DurationSymbolTraining} = \qty{1}{\mega\hertz}\), and \(\gls{NoisePower} = \gls{NoisePowerSpectralDensity}\gls{NoiseFigure}\gls{BandwidthPilot}\) with noise power spectral density \(\gls{NoisePowerSpectralDensity} = \qty{-174}{\dBm/\hertz}\) and noise figure \(\gls{NoiseFigure}=\qty{6}{\deci\bel}\).
Additional parameters are individually given for each figure.

\subsection{SNR Regimes}\label{sec:numerical_results_snr_regimes}
\begin{figure}[!t]
  \centering
  \includegraphics{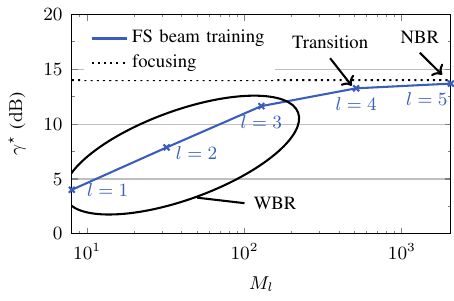}
  \caption{\Gls{snr} in \eqref{eq:snr-data} as a function of codebook size for \(\gls{NumSymTxLevel} = 8\) and \(\gls{Velocity} = \qty{3}{\kilo\metre/\hour}\). The corresponding levels of a hierarchical codebook are indicated by cross markers.}\label{fig:snr_verification}
\end{figure}
\noindent
We first evaluate the impact of the codebook size on the achievable \gls{snr} by comparing the numerical results for beam training and with that for focusing.
In Fig.~\ref{fig:snr_verification}, we show the results for \gls{fs} beam training as a function of the codebook size, and indicate the corresponding codebook levels of a hierarchical codebook by cross markers.
One can see that \(\gls{SnrReceivedData}\) increases with \(\gls{SizeCodebookAtLevel}\) if the codebook is small but saturates at the upper bound for large \(\gls{SizeCodebookAtLevel}\).
The reason is that, for a given coverage area, the size of the subareas is inversely proportional to the codebook size, leading to a change of \gls{snr} regime as \(\gls{SizeCodebookAtLevel}\) grows.
For example, note that the \gls{snr} achieved by beam training is similar to that for focusing if \(\gls{SizeCodebookAtLevel} > 500\).
In this case, the subareas are sufficiently small to be approximately illuminated by narrow beams.
However, a lower \gls{snr} is observed for \(\gls{SizeCodebookAtLevel} < 500\) because wide beams are required to distribute the power collected by the \gls{ris} to larger subareas.
These results accurately match the scaling laws in Proposition~\ref{prop:snr-scaling-law}, which states that the \gls{snr} linearly scales with \(\gls{SizeCodebookAtLevel}\) in the \gls{wbr} and is independent of \(\gls{SizeCodebookAtLevel}\) in the \gls{nbr}.
Moreover, the markers Fig.~\ref{fig:snr_verification} show that the system operates in different \gls{snr} regimes during \gls{hs} beam training with \(\gls{NumLevels} = 5\) levels.
In particular, the \gls{snr} is in the \gls{wbr}, the transition regime, and the \gls{nbr} for \(l \leq 3\), \(l=4\), and \(l=5\), respectively, which confirms the analytical results shown in Fig.~\ref{fig:footprints_subareas_regimes}.

\subsection{Overhead Regimes}
\begin{figure}[!t]
  \centering
    \includegraphics{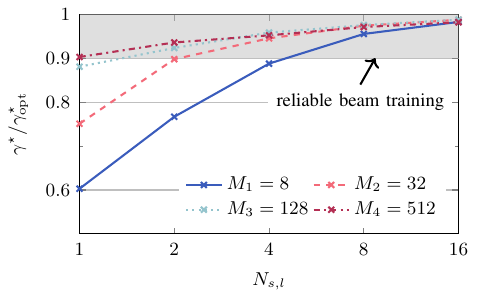}
  \caption{Beam training reliability as a function of the number of pilot symbols for \(\gls{Velocity} = \qty{3}{\kilo\metre/\hour}\) and different codebook sizes corresponding to \(\gls{NumLevels} = 4\) codebook levels.}\label{fig:pilots}
\end{figure}
\noindent
In order to evaluate the impact of the number of pilot symbols on the \gls{snr} achieved by beam training, we adopt the ratio \(\gls{SnrReceivedData}/\gls{SnrReceivedDataCodebookBound}\) as a measure for the beam training reliability, where \(\gls{SnrReceivedDataCodebookBound}\) denotes the \gls{snr} for data transmission achieved by optimal codeword selection.
Thus, \(\gls{SnrReceivedDataCodebookBound}\) represents the \gls{snr} that can be achieved without the impact of channel fading or receiver noise.
Moreover, we assume that beam training can be considered reliable if \(\gls{SnrReceivedData}/\gls{SnrReceivedDataCodebookBound} > 0.9\).
For the considered system parameters, Fig.~\ref{fig:pilots} shows that reliable beam training is only achieved when \(\gls{NumSymTxLevel}\) or \(\gls{SizeCodebookAtLevel}\) is large.
Otherwise, the beam training reliability is significantly reduced, which can be explained by a larger number of erroneous codeword selections owing to a low \gls{snr}.
These observations are consistent with \eqref{eq:pilot_symbol_bound} and the definitions of \gls{ior} and \gls{dor} in Section~\ref{sec:regimes-overhead}.
For example, the numerical results in Fig.~\ref{fig:pilots} reveal that \gls{ior} holds for \(\gls{SizeCodebookAtLevel} = 512\) because \(\gls{SnrReceivedData}/\gls{SnrReceivedDataCodebookBound} > 0.9\) is achieved with \(\gls{NumSymTxLevel} = 1\).
In contrast, \gls{dor} holds for \(\gls{SizeCodebookAtLevel} \leq 128\) where \(\gls{NumSymTxLevel} > 1\) is required to obtain \(\gls{SnrReceivedData}/\gls{SnrReceivedDataCodebookBound} > 0.9\).

\subsection{Impact of RIS Response Time and User Velocity}
\begin{figure}[!t]
  \centering
  \includegraphics{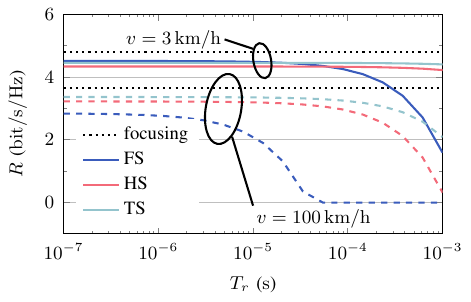}
  \caption{Effective rate as a function of the \gls{ris} response time for \(\gls{NumSymTxLevel} = 8\), \(\gls{NumLevels} = 4\), \(\gls{Velocity} = \qty{3}{\kilo\metre/\hour}\) (solid lines), and \(\gls{Velocity} = \qty{100}{\kilo\metre/\hour}\) (dashed lines).}\label{fig:order_typical}
\end{figure}
\noindent
Since our theoretical analysis has shown that both the \gls{ris} response time and the user velocity can have a significant impact on the system performance, we evaluate the effective rate for different values of \(\gls{TimeResponse}\) and \(\gls{Velocity}\) in Fig.~\ref{fig:order_typical}.
For \(\gls{Velocity} = \qty{3}{\kilo\metre/\hour}\), one can see that all considered beam training strategies provide a similar effective rate as focusing if \(\gls{TimeResponse} < \qty{100}{\micro\second}\), and only \gls{fs} beam training results in reduced rates if \(\gls{TimeResponse} > \qty{100}{\micro\second}\).
In contrast, one observes significant performance differences for \(\gls{Velocity} = \qty{100}{\kilo\metre/\hour}\).
These results suggest that, for a slowly moving user, the beam training overhead is negligible if the \gls{ris} response time is less than \qty{100}{\micro\second}, which can be realized with \glspl{ris} based on PIN diodes or MEMS, cf. Table~\ref{tab:technology}.
In this case, \gls{fs} is the preferred beam training strategy because decision and tracking errors may reduce the \gls{snr} achieved by \gls{hs} and \gls{ts} beam training, respectively.
Although not significant, this effect can be seen in Fig.~\ref{fig:order_typical} for \(\gls{TimeResponse} < \qty{10}{\micro\second}\) and \(\gls{Velocity} = \qty{3}{\kilo\metre/\hour}\), where the effective rates for \gls{hs} and \gls{ts} beam training are slightly lower than that for \gls{fs} beam training.
However, \gls{fs} beam training provides the worst performance if \(\gls{Velocity} = \qty{100}{\kilo\metre/\hour}\).
In this case, Fig.~\ref{fig:order_typical} demonstrates that low-overhead beam training is essential for achieving high effective rates, especially for \gls{ris} technologies with a slow response time in the order of milliseconds, e.g., based on liquid crystal.

\begin{figure}[!t]
  \centering
  \includegraphics{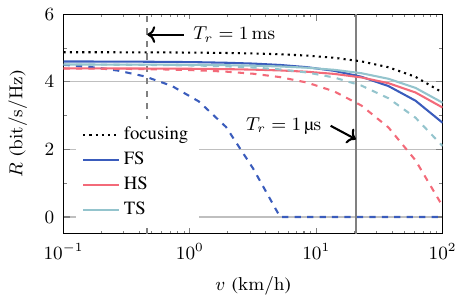}
  \caption{Effective rate as a function of the user velocity for \(\gls{NumLevels} = 4\), \(\gls{TimeResponse} = \qty{1}{\micro\second}\) (solid lines), and \(\gls{TimeResponse} = \qty{1}{\milli\second}\) (dashed lines). The values of the bound \(\gls{VelocityBound}\) in \eqref{eq:velocity_bound} for \(\gls{OverheadThreshold} = 0.1\) are indicated by the vertical lines.}\label{fig:relevance}
\end{figure}

More insight on the impact of the overhead on the system performance is provided in Fig.~\ref{fig:relevance}, which shows the effective rate as a function of the user velocity for \(\gls{TimeResponse} \in \{\qty{1}{\micro\second}, \qty{1}{\milli\second}\}\).
One can see that a fast \gls{ris} with \(\gls{TimeResponse} = \qty{1}{\micro\second}\) (solid lines in Fig.~\ref{fig:relevance}) enables beam training for high user velocities without significant performance loss, in particular, if \gls{hs} or \gls{ts} beam training is applied.
In contrast, employing a slow \gls{ris} with \(\gls{TimeResponse} = \qty{1}{\milli\second}\) reduces the performance for all considered beam training strategies.
The most significant impact is observed for \gls{fs} beam training, where the largest velocity that has negligible impact on \(\gls{RateEffective}\) decreases from about \(\qty{20}{\kilo\metre/\hour}\) to \(\qty{0.4}{\kilo\metre/\hour}\).
As indicated by the vertical lines in Fig.~\ref{fig:relevance}, these results accurately match the analytical bound \(\gls{VelocityBound}\) in \eqref{eq:velocity_bound} for \(\gls{OverheadThreshold} = 0.1\).

\subsection{Impact of Feedback Delay}
\noindent
Finally, we evaluate the effective rate as a function of the feedback delay in Fig.~\ref{fig:order_special}, considering a fast \gls{ris} with response times \(\gls{TimeResponse} = \qty{1}{\micro\second}\) (solid lines) and \(\gls{TimeResponse} = \qty{30}{\micro\second}\) (dashed lines).
The results show that the effective rates are approximately constant for \(\gls{DelayFeedback} < \qty{1}{\milli\second}\), but decrease for larger values of \(\gls{DelayFeedback}\).
Moreover, one can see that \gls{ts} provides the best beam training performance, whereas \gls{hs} beam training outperforms \gls{fs} beam training or vice versa, depending on the values of \(\gls{DelayFeedback}\) and \(\gls{TimeResponse}\).
This is consistent with our theoretical results of Lemma~\ref{lm:ts_superior} and can be explained by the fact that \gls{hs} beam training involves multiple feedback transmissions, which cause a large overhead if \(\gls{DelayFeedback}\) is large.
Thus, if the \gls{ris} response time is relatively short, \gls{fs} beam training can have a lower overhead than \gls{hs} beam training, leading to a larger effective rate.
In Fig.~\ref{fig:order_special}, this happens for \(\gls{DelayFeedback} > \qty{0.9}{\milli\second}\) (\(\gls{TimeResponse} = \qty{1}{\micro\second}\)) and \(\gls{DelayFeedback} > \qty{6}{\milli\second}\) (\(\gls{TimeResponse} = \qty{30}{\micro\second}\)).
As indicated by the vertical lines, these results are accurately captured by the lower bound \(\gls{DelayFeedbackBound}\) in \eqref{eq:condition_fs_better_hs}.

\begin{figure}[!t]
  \centering
  \includegraphics{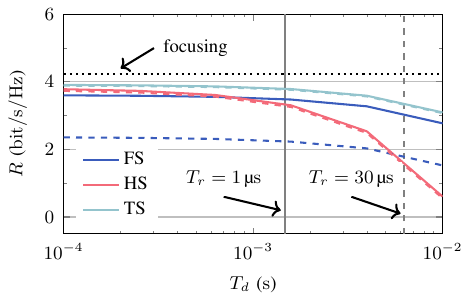}
  \caption{Effective rate as a function of the feedback delay for \(\gls{NumLevels} = 4\), \(\gls{Velocity} = \qty{50}{\kilo\metre/\hour}\), \(\gls{TimeResponse} = \qty{1}{\micro\second}\) (solid lines), and \(\gls{TimeResponse} = \qty{30}{\micro\second}\) (dashed lines). The values of the bound \(\gls{DelayFeedbackBound}\) in \eqref{eq:condition_fs_better_hs} are indicated by the vertical lines.}
  \label{fig:order_special}
\end{figure}

\section{Conclusion}\label{sec:conclusion}
\noindent
This paper studied the performance tradeoff between the overhead and achievable \gls{snr} in \gls{ris} beam training, considering the size of the coverage, the \gls{ris} response time, and the feedback delay.
In particular, we derived scaling laws for the achievable \gls{snr} based on narrow and wide beam designs, and showed that the overhead for reliable beam training is either dependent on or independent of the \gls{snr}.
Based on these insights, we investigated the impact of the overhead for \gls{fs}, \gls{hs}, and \gls{ts} beam training on the effective rate, and provided an upper bound on the user velocity for which the overhead is negligible.
Moreover, when the overhead is not negligible, we showed that \gls{ts} beam training achieves higher effective rates than \gls{hs} and \gls{fs} beam training, while we revealed that \gls{hs} beam training may or may not outperform \gls{fs} beam training.
Furthermore, our theoretical results were verified by numerical simulations.
In particular, our numerical results demonstrated that fast \glspl{ris} facilitate \gls{fs} beam training, whereas large feedback delays can significantly reduce the performance for \gls{hs} beam training.

\appendix
\section*{}
\subsection{SNR Scaling Laws}\label{sec:proof-snr-scaling}
\noindent
In this section, we derive the scaling laws of the achievable \gls{snr} for wide beams and narrow beams, respectively.

\subsubsection{Wide-Beam Design}\label{sec:proof-snr-scaling-wbr}
\noindent
Assuming free-space propagation and a precoding vector \(\gls{PrecodingAp}\) that aligns the main lobe of the \gls{bs} with the \gls{los} path to the \gls{ris}, a signal transmitted at the \gls{bs} causes power density
\(
  \gls{PowerDensityIncident}
  = \frac{\gls{PowerTransmitter} \gls{GainTransmitter} \gls{NumAp}}{4 \pi \gls{DistanceIncident}^2}
\)
at the \gls{ris}, where \(\gls{DistanceIncident}\) denotes the distance between the \gls{bs} and the \gls{ris}.
Moreover, the effective area of the \gls{ris} is given by \(\gls{AreaRisEffective} = \gls{LengthUnitCellX} \gls{LengthUnitCellY} \gls{NumRis} \gls{AreaRisFactor}\), where \(\gls{AreaRisFactor}\) accounts for the \gls{aoa} and polarization of the incident wave~\cite{najafi2020physicsbasedmodeling,jamali2022lowzerooverhead}.
Thus, the power collected by the \gls{ris} is given by
\(
  \gls{PowerWaveIncident} = \gls{PowerDensityIncident} \gls{LengthUnitCellX} \gls{LengthUnitCellY} \gls{NumRis} \gls{AreaRisFactor}
\)~\cite{laue2022performancetradeoffris}.
Furthermore, for the idealized wide-beam design and subareas of equal size, the \(m\)th codeword at the \(l\)th codebook level perfectly distributes the collected power \(\gls{PowerWaveIncident}\) across an area of size \(\gls{Area}/\gls{SizeCodebookAtLevel}\), where \(\gls{Area}/\gls{SizeCodebookAtLevel}\) is larger than \(\gls{PathThroughFootprint}\), i.e., the beam footprint that results from focusing on one single direction.
Then, the power density at a location within the subarea is given by
\(
  \gls{PowerDensityReflected} = \gls{PowerWaveIncident} \frac{\gls{SizeCodebookAtLevel}}{\gls{Area}}
\)~\cite{laue2022performancetradeoffris},
which leads to achievable \gls{snr}
\begin{equation}\label{eq:snr-normalized-wbr}
  \gls{SnrReceivedNormalized}
    = \frac{\gls{PowerDensityReflected} \gls{AreaReceiveAntennaEffective}}{\gls{NoisePower}}
    = \frac{\gls{PowerTransmitter} \gls{GainTransmitter} \gls{NumAp} \gls{GainReceiver}}{\gls{NoisePower}}
    \frac{\glsSquared{Wavelength} \gls{LengthUnitCellX} \gls{LengthUnitCellY} \gls{AreaRisFactor}}{(4 \pi \gls{DistanceIncident})^2 \gls{Area}} \gls{NumRis} \gls{SizeCodebookAtLevel},
\end{equation}
where \(\gls{AreaReceiveAntennaEffective} = \glsSquared{Wavelength} \gls{GainReceiver}/(4 \pi)\) denotes the effective receive antenna area.

\subsubsection{Narrow-Beam Design}\label{sec:proof-snr-scaling-nbr}
\noindent
A codebook with narrow beams corresponds to sampling the coverage area in the angular space, where the \(m\)th codeword at the \(l\)th codebook level results in a narrow beam focusing on a particular direction \((\theta_{m,l}, \phi_{m,l})\).
Here, we assume that \((\theta_{m,l}, \phi_{m,l})\) denotes the direction to the center of the \(m\)th subarea for the \(l\)th codebook level.
Then, the achievable \gls{snr} is given by~\cite{najafi2020physicsbasedmodeling}
\begin{equation}\label{eq:snr-nbr}
  \gls{SnrReceivedNormalized} = \frac{\gls{PowerTransmitter} \gls{GainTransmitter} \gls{NumAp} \gls{GainReceiver}}{\gls{NoisePower}}
  \frac{\glsSquared{Wavelength} \glsSquared{Wavelength} \abs*{\gls{GainRisWithAngles}}^2}{\left(4 \pi \gls{DistanceIncident}\right)^2 \left(4 \pi \gls{DistanceReflected}\right)^2},
\end{equation}
where \(\gls{DistanceReflected}\) and \(\abs*{\gls{GainRisWithAngles}}^2\) denote the \gls{ris}-user distance and the \gls{ris} gain for the \(m\)th codeword at the \(l\)th codebook level, respectively.
For a user located in direction \((\theta_r, \phi_r)\), the \gls{ris} gain is given by~\cite{najafi2020physicsbasedmodeling}
\begin{equation}\label{eq:ris_response_magnitude}
  \abs*{\gls{GainRisWithAngles}}^2 = \abs*{\gls{GainUc}}^2 \abs*{\frac{\sin(\sqrt{Q} W_x)}{\sin(W_x)} \frac{\sin(\sqrt{Q} W_y)}{\sin(W_y)}}^2,
\end{equation}
where \(W_o = \frac{\pi d_o}{\gls{Wavelength}}(w_o(\theta_r, \phi_r) - w_o(\theta_{m,l}, \phi_{m,l}))\), \(o \in \{x,y\}\), for \(w_x(\theta, \phi) = \sin(\theta)\cos(\phi)\) and \(w_y(\theta, \phi) = \sin(\theta)\sin(\phi)\).
Moreover, if the subareas of the coverage area are significantly smaller than the footprint of a narrow beam, one can always find a codeword that approximately provides the maximum \gls{ris} gain for any direction \((\theta_r, \phi_r)\).
Since \(\abs*{\gls{GainRisWithAngles}}^2 \leq \abs*{\gls{GainUc}}^2 \glsSquared{NumRis}\)~\cite{najafi2020physicsbasedmodeling}, the \gls{snr} in \eqref{eq:snr-nbr} can be written as
\begin{equation}\label{eq:snr-normalized-nbr}
  \gls{SnrReceivedNormalized} \approx \frac{\gls{PowerTransmitter} \gls{GainTransmitter} \gls{NumAp} \gls{GainReceiver}}{\gls{NoisePower}}
  \frac{\glsSquared{Wavelength} \glsSquared{Wavelength} \abs*{\gls{GainUc}}^2 \glsSquared{NumRis}}{\left(4 \pi \gls{DistanceIncident}\right)^2 \left(4 \pi \gls{DistanceReflected}\right)^2}.
\end{equation}

\subsection{Proof of Lemma~\ref{lm:negligible_overhead}}\label{sec:proof-negligible-overhead}
\noindent
We need to show that a user velocity
\(
  \gls{Velocity} > \frac{
  \gls{ParabolaParameterOne} + \gls{ParabolaParameterTwo}
  + \sqrt{
    (\gls{ParabolaParameterOne}+\gls{ParabolaParameterTwo})^2
    - 4\gls{ParabolaParameterOne}\gls{ParabolaParameterTwo}\gls{OverheadThreshold}
  }
  }{
    2\gls{ParabolaParameterOne}\gls{ParabolaParameterTwo}
  }
\) results in zero time for data transmission.
For the considered transmission protocol, cf. Fig.~\ref{fig:protocol}, this case is true for \(\gls{DurationFrame} \leq \gls{OverheadTraining} + \gls{OverheadEstimation}\), which is equivalent to \(\frac{1}{\gls{ParabolaParameterOne} (1 + \gls{OverheadEstimation}/\gls{OverheadTraining})} \leq \gls{Velocity}\).
Thus, the proof is completed if we show that
\(
  \frac{1}{\gls{ParabolaParameterOne} (1 + \gls{OverheadEstimation}/\gls{OverheadTraining})} \leq \frac{
    \gls{ParabolaParameterOne} + \gls{ParabolaParameterTwo}
    + \sqrt{
      (\gls{ParabolaParameterOne}+\gls{ParabolaParameterTwo})^2
      - 4\gls{ParabolaParameterOne}\gls{ParabolaParameterTwo}\gls{OverheadThreshold}
    }
  }{
    2\gls{ParabolaParameterOne}\gls{ParabolaParameterTwo}
  }
\) holds.
Since \(\gls{OverheadTraining} \geq 0\), \(\gls{OverheadEstimation} \geq 0\), and \(\gls{OverheadThreshold} < 1\), the above condition is certainly true if
\(
  \gls{ParabolaParameterOne} - \gls{ParabolaParameterTwo}
  + \sqrt{
    (\gls{ParabolaParameterOne}+\gls{ParabolaParameterTwo})^2
    - 4\gls{ParabolaParameterOne}\gls{ParabolaParameterTwo}
  } \geq 0
\).
For \(\gls{ParabolaParameterOne} > 0\) and \(\gls{ParabolaParameterTwo} > 0\), this can be rewritten as
\(
  \gls{ParabolaParameterOne} - \gls{ParabolaParameterTwo}
  + \abs*{\gls{ParabolaParameterOne} - \gls{ParabolaParameterTwo}} \geq 0
\), which is true.

\subsection{Proof of Corollary~\ref{cr:condition_no_impact}}\label{sec:proof-overhead-bound}
\noindent
In order to find an upper bound
\(
  \gls{OverheadTrainingBound}
  \geq \max_{\mathsf{x} \in \{\text{\gls{fs}},\text{\gls{hs}},\text{\gls{ts}}\}} \gls{OverheadTrainingX}
\), we note that \gls{dor} and \gls{wbr} cause the largest overhead because they require a larger number of pilot symbols compared to \gls{ior} and \gls{nbr}.
Moreover, a comparison of \eqref{eq:overhead_training_full} and \eqref{eq:overhead_training_tracking} shows that the overhead for \gls{fs} beam training is always larger than that for \gls{ts} beam training if \(\gls{SizeCodebookAtLevel} > 8\) is assumed.
Thus, it is sufficient to compare the overheads for \gls{fs} beam training and \gls{hs} beam training.

The considered \gls{hs} beam training is based on a codebook that grows by a factor of \(\gls{NumCodewordsFactor}\) with each level.
Thus, in the \gls{wbr}, the \gls{snr} increases by \(\gls{NumCodewordsFactor}\) at each codebook level, which implies that the number of required pilot symbols decreases as \(\gls{NumSymTxLevelNext} = \gls{NumSymTxLevel}/\gls{NumCodewordsFactor}\).
Rewriting the latter as \(\gls{NumSymTxLevel} = \gls{NumCodewordsFactorScaling} \gls{NumSymTxLevelMax}\) leads to an upper bound for \(\gls{OverheadTrainingHierarchical}\) as follows
\begin{subequations}
  \begin{align}
    \gls{OverheadTrainingHierarchical}
    &= \gls{TimeResponse} \left(1 + \gls{NumCodewordsTotalHierarchical}\right)
    + \gls{DelayFeedback} \gls{NumLevels}
    + \gls{DurationSymbolTraining}
    \sum_{l=1}^{\gls{NumLevels}} \gls{NumSymTxLevel} \gls{NumCodewordsTraining}\\
    &\overset{(a)}{=} \gls{TimeResponse} \left(1 + \gls{NumCodewordsTotalHierarchical}\right)
    + \gls{DelayFeedback} \gls{NumLevels}
    + \gls{DurationSymbolTraining} \gls{NumSymTxLevelMax} \left(
      \gls{SizeCodebookAtLevelMax}
      + \frac{\gls{NumCodewordsFactorPowerHigh} - 1}{1 - 1/\gls{NumCodewordsFactor}}
    \right)\\
    &\overset{(b)}{<} \gls{TimeResponse} \left(1 + \gls{NumCodewordsTotalHierarchical}\right)
    + \gls{DelayFeedback} \gls{NumLevels}
    + \gls{DurationSymbolTraining} \gls{NumSymTxLevelMax} \gls{SizeCodebookAtLevelMax} 1.25,\label{eq:overhead_training_hierarchical_bound}
  \end{align}
\end{subequations}
where \((a)\) is based on the partial sum of a geometric series and \((b)\) results from \(\gls{NumCodewordsFactor} \geq 2\) and \(\gls{SizeCodebookAtZero} > 8\).
Then, combining \eqref{eq:overhead_training_hierarchical_bound} and \(\gls{OverheadTrainingFull} = \gls{DurationSymbolTraining} \gls{NumSymTxLevelMax} \gls{SizeCodebookAtLevelMax}
+ \gls{TimeResponse} \left(1 + \gls{SizeCodebookAtLevelMax}\right)
+ \gls{DelayFeedback}\) yields
\begin{equation}\label{eq:overhead_upper_bound}
  \gls{OverheadTrainingBound}
  = 1.25 \gls{NumSymTxLevelMax} \gls{SizeCodebookAtLevelMax} \gls{DurationSymbolTraining}
  + \left(1 + \gls{SizeCodebookAtLevelMax}\right) \gls{TimeResponse}
  + \gls{NumLevels} \gls{DelayFeedback}.
\end{equation}

\bibliographystyle{IEEEtran}
\bibliography{IEEEabrv,bibtexlib}

\begin{thebibliography}{10}
\providecommand{\url}[1]{#1}
\csname url@samestyle\endcsname
\providecommand{\newblock}{\relax}
\providecommand{\bibinfo}[2]{#2}
\providecommand{\BIBentrySTDinterwordspacing}{\spaceskip=0pt\relax}
\providecommand{\BIBentryALTinterwordstretchfactor}{4}
\providecommand{\BIBentryALTinterwordspacing}{\spaceskip=\fontdimen2\font plus
\BIBentryALTinterwordstretchfactor\fontdimen3\font minus
  \fontdimen4\font\relax}
\providecommand{\BIBforeignlanguage}[2]{{%
\expandafter\ifx\csname l@#1\endcsname\relax
\typeout{** WARNING: IEEEtran.bst: No hyphenation pattern has been}%
\typeout{** loaded for the language `#1'. Using the pattern for}%
\typeout{** the default language instead.}%
\else
\language=\csname l@#1\endcsname
\fi
#2}}
\providecommand{\BIBdecl}{\relax}
\BIBdecl

\bibitem{laue2022performancetradeoffris}
F.~Laue, M.~Garkisch, V.~Jamali, and R.~Schober, ``Performance tradeoff of
  {RIS} beam training: Overhead vs. achievable {SNR},'' in \emph{Proc. {IEEE}
  56th Asilomar Conf. Signals, Systems, and Comput.}, Oct. 2022.

\bibitem{renzo2020smartradioenvironments}
M.~D. Renzo, A.~Zappone, M.~Debbah, M.-S. Alouini, C.~Yuen, J.~de~Rosny, and
  S.~Tretyakov, ``Smart radio environments empowered by reconfigurable
  intelligent surfaces: How it works, state of research, and the road ahead,''
  \emph{{IEEE} J. Sel. Areas Commun.}, vol.~38, no.~11, pp. 2450--2525, Jul.
  2020.

\bibitem{gong2020smartwirelesscommunications}
S.~Gong, X.~Lu, D.~T. Hoang, D.~Niyato, L.~Shu, D.~I. Kim, and Y.-C. Liang,
  ``Toward smart wireless communications via intelligent reflecting surfaces: A
  contemporary survey,'' \emph{{IEEE} Commun. Surveys \& Tut.}, vol.~22, no.~4,
  pp. 2283--2314, fourthquarter 2020.

\bibitem{wu2021intelligentreflectingsurface}
Q.~Wu, S.~Zhang, B.~Zheng, C.~You, and R.~Zhang, ``Intelligent reflecting
  surface-aided wireless communications: A tutorial,'' \emph{IEEE Trans.
  Commun.}, vol.~69, no.~5, pp. 3313--3351, May 2021.

\bibitem{pan2022overviewsignalprocessing}
C.~Pan, G.~Zhou, K.~Zhi, S.~Hong, T.~Wu, Y.~Pan, H.~Ren, M.~D. Renzo, A.~L.
  Swindlehurst, R.~Zhang, and A.~Y. Zhang, ``An overview of signal processing
  techniques for {RIS}/{IRS}-aided wireless systems,'' \emph{IEEE J. Sel.
  Topics Signal Process.}, vol.~16, no.~5, pp. 883--917, Aug. 2022.

\bibitem{zheng2022surveychannelestimation}
B.~Zheng, C.~You, W.~Mei, and R.~Zhang, ``A survey on channel estimation and
  practical passive beamforming design for intelligent reflecting surface aided
  wireless communications,'' \emph{{IEEE} Commun. Surveys \& Tut.}, vol.~24,
  no.~2, pp. 1035--1071, secondquarter 2022.

\bibitem{an2022codebookbasedsolutions}
J.~An, C.~Xu, Q.~Wu, D.~W.~K. Ng, M.~D. Renzo, C.~Yuen, and L.~Hanzo,
  ``Codebook-based solutions for reconfigurable intelligent surfaces and their
  open challenges,'' \emph{IEEE Wirel. Commun.}, Nov. 2022.

\bibitem{najafi2020physicsbasedmodeling}
M.~Najafi, V.~Jamali, R.~Schober, and H.~V. Poor, ``Physics-based modeling and
  scalable optimization of large intelligent reflecting surfaces,''
  \emph{{IEEE} Trans. Commun.}, vol.~69, no.~4, pp. 2673--2691, Dec. 2020.

\bibitem{liu2022simulationfieldtrial}
R.~Liu, J.~Dou, P.~Li, J.~Wu, and Y.~Cui, ``Simulation and field trial results
  of reconfigurable intelligent surfaces in {5G} networks,'' \emph{{IEEE}
  Access}, vol.~10, pp. 122\,786--122\,795, 2022.

\bibitem{peng2022channelestimationris}
Z.~Peng, G.~Zhou, C.~Pan, H.~Ren, A.~L. Swindlehurst, P.~Popovski, and G.~Wu,
  ``Channel estimation for {RIS}-aided multi-user {mmWave} systems with uniform
  planar arrays,'' \emph{IEEE Trans. Commun.}, vol.~70, no.~12, pp. 8105--8122,
  Dec. 2022.

\bibitem{wang2023hierarchicalcodebookbased}
J.~Wang, W.~Tang, S.~Jin, C.-K. Wen, X.~Li, and X.~Hou, ``Hierarchical
  codebook-based beam training for {RIS}-assisted {mmWave} communication
  systems,'' \emph{IEEE Trans. Commun.}, vol.~71, no.~6, pp. 3650--3662, Jun.
  2023.

\bibitem{jamali2022lowzerooverhead}
V.~Jamali, G.~C. Alexandropoulos, R.~Schober, and H.~V. Poor,
  ``Low-to-zero-overhead {IRS} reconfiguration: Decoupling illumination and
  channel estimation,'' \emph{{IEEE} Commun. Lett.}, vol.~26, no.~4, pp.
  932--936, Apr. 2022.

\bibitem{zhang2022dualcodebookdesign}
Y.~Zhang, B.~Di, H.~Zhang, M.~Dong, L.~Yang, and L.~Song, ``Dual codebook
  design for intelligent omni-surface aided communications,'' \emph{IEEE Trans.
  Wireless Commun.}, vol.~21, no.~11, pp. 9232--9245, Nov. 2022.

\bibitem{zhang2023rateoverheadtradeoff}
S.~Zhang, Y.~Zhang, H.~Zhang, N.~Ye, and B.~Di, ``Rate-overhead tradeoff for
  {IOS}-aided beam training: How large codebook is enough for the {IOS}?''
  \emph{IEEE Wireless Commun. Lett.}, vol.~12, no.~6, pp. 1081--1085, Jun.
  2023.

\bibitem{zhang2022rateoverheadtradeoff}
S.~Zhang, Y.~Zhang, B.~Di, and H.~Zhang, ``Rate-overhead tradeoff in beam
  training for {RRS}-assisted multi-user communications,'' in \emph{Proc.
  {IEEE} 96th Veh. Technol. Conf.}, Sep. 2022.

\bibitem{hamidisepehr20215gurllcevolution}
F.~Hamidi-Sepehr, M.~Sajadieh, S.~Panteleev, T.~Islam, I.~Karls, D.~Chatterjee,
  and J.~Ansari, ``{5G} {URLLC}: Evolution of high-performance wireless
  networking for industrial automation,'' \emph{{IEEE} Commun. Standards Mag.},
  vol.~5, no.~2, pp. 132--140, Jun. 2021.

\bibitem{ji2021severalkeytechnologies}
B.~Ji, Y.~Han, S.~Liu, F.~Tao, G.~Zhang, Z.~Fu, and C.~Li, ``Several key
  technologies for {6G}: Challenges and opportunities,'' \emph{{IEEE} Commun.
  Standards Mag.}, vol.~5, no.~2, pp. 44--51, Jun. 2021.

\bibitem{jimenezsaez2023reconfigurableintelligentsurfaces}
\BIBentryALTinterwordspacing
A.~Jiménez-Sáez, A.~Asadi, R.~Neuder, M.~Delbari, and V.~Jamali,
  ``Reconfigurable intelligent surfaces with liquid crystal technology: A
  hardware design and communication perspective,'' Aug. 2023. [Online].
  Available: \url{https://arxiv.org/abs/2308.03065}
\BIBentrySTDinterwordspacing

\bibitem{wang2022beamtrainingalignment}
P.~Wang, J.~Fang, W.~Zhang, Z.~Chen, H.~Li, and W.~Zhang, ``Beam training and
  alignment for {RIS}-assisted millimeter-wave systems: State of the art and
  beyond,'' \emph{IEEE Wirel. Commun.}, vol.~29, no.~6, pp. 64--71, Dec. 2022.

\bibitem{kamoda201160ghzelectronically}
H.~Kamoda, T.~Iwasaki, J.~Tsumochi, T.~Kuki, and O.~Hashimoto, ``60-{GHz}
  electronically reconfigurable large reflectarray using single-bit phase
  shifters,'' \emph{IEEE Trans. Antennas Propag.}, vol.~59, no.~7, pp.
  2524--2531, Jul. 2011.

\bibitem{yang20161bit10}
H.~Yang, F.~Yang, S.~Xu, Y.~Mao, M.~Li, X.~Cao, and J.~Gao, ``A 1-bit 10x10
  reconfigurable reflectarray antenna: Design, optimization, and experiment,''
  \emph{IEEE Trans. Antennas Propag.}, vol.~64, no.~6, pp. 2246--2254, Jun.
  2016.

\bibitem{li2021programmablemetasurfacebased}
Y.~Li, J.~Eisenbeis, X.~Wan, S.~Bettinga, X.~Long, M.~B. Alabd, J.~Kowalewski,
  T.~Cui, and T.~Zwick, ``A programmable-metasurface-based {TDMA} fast beam
  switching communication system at 28~{GHz},'' \emph{IEEE Antennas Wirel.
  Propag. Lett.}, vol.~20, no.~5, pp. 658--662, May 2021.

\bibitem{pan202110240element}
X.~Pan, F.~Yang, S.~Xu, and M.~Li, ``A 10 240-element reconfigurable
  reflectarray with fast steerable monopulse patterns,'' \emph{IEEE Trans.
  Antennas Propag.}, vol.~69, no.~1, pp. 173--181, Jan. 2021.

\bibitem{sharma2017designsimulationcomparative}
N.~Sharma, R.~Smitha, D.~Kumar, and A.~S. Siddiqui, ``Design, simulation and a
  comparative study of square, rectangular, triangular and dual beamed {RF}
  {MEMS} switch for switching applications,'' in \emph{Proc. {IEEE} 4th Int.
  Conf. Signal Process. Integr. Netw.}, Feb. 2017.

\bibitem{sravani2019designperformanceanalysis}
K.~G. Sravani, D.~Prathyusha, K.~S. Rao, P.~A. Kumar, G.~S. Lakshmi, C.~G.
  Chand, P.~Naveena, L.~N. Thalluri, and K.~Guha, ``Design and performance
  analysis of low pull-in voltage of dimple type capacitive {RF} {MEMS} shunt
  switch for ka-band,'' \emph{{IEEE} Access}, vol.~7, pp. 44\,471--44\,488,
  2019.

\bibitem{nooraiyeen2020designnovellow}
A.~Nooraiyeen, G.~Priya, N.~K. Kavya, V.~Sanchitha, and S.~B. Rudraswamy,
  ``Design of novel low actuation voltage shunt capacitive {RF} {MEMS}
  switch,'' in \emph{Proc. {IEEE} Int. Conf. Smart Electron. Commun.}, Sep.
  2020.

\bibitem{jakoby2020microwaveliquidcrystal}
R.~Jakoby, A.~Gaebler, and C.~Weickhmann, ``Microwave liquid crystal enabling
  technology for electronically steerable antennas in {SATCOM} and {5G}
  millimeter-wave systems,'' \emph{Crystals}, vol.~10, no.~6, p. 514, Jun.
  2020.

\bibitem{jamali2021powerefficiencyoverhead}
V.~Jamali, M.~Najafi, R.~Schober, and H.~V. Poor, ``Power efficiency, overhead,
  and complexity tradeoff of {IRS} codebook design - quadratic phase-shift
  profile,'' \emph{{IEEE} Commun. Lett.}, vol.~25, no.~6, pp. 2048--2052, Jun.
  2021.

\bibitem{ghanem2022optimizationbasedphasea}
W.~R. Ghanem, V.~Jamali, M.~Schellmann, H.~Cao, J.~Eichinger, and R.~Schober,
  ``Optimization-based phase-shift codebook design for large {IRSs},''
  \emph{IEEE Commun. Lett.}, vol.~27, no.~2, pp. 635--639, Feb. 2023.

\bibitem{lv2023risaidedfield}
\BIBentryALTinterwordspacing
S.~Lv, Y.~Liu, X.~Xu, A.~Nallanathan, and A.~L. Swindlehurst, ``{RIS}-aided
  near-field {MIMO} communications: Codebook and beam training design,'' Sep.
  2023. [Online]. Available: \url{https://arxiv.org/abs/2310.00294}
\BIBentrySTDinterwordspacing

\bibitem{you2020fastbeamtraining}
C.~You, B.~Zheng, and R.~Zhang, ``Fast beam training for {IRS}-assisted
  multiuser communications,'' \emph{IEEE Wireless Commun. Lett.}, vol.~9,
  no.~11, pp. 1845--1849, Nov. 2020.

\bibitem{alexandropoulos2022fieldhierarchicalbeama}
G.~C. Alexandropoulos, V.~Jamali, R.~Schober, and H.~V. Poor, ``Near-field
  hierarchical beam management for {RIS}-enabled millimeter wave multi-antenna
  systems,'' in \emph{Proc. {IEEE} 12th Sensor Array and Multichannel Signal
  Process. Workshop}, Jun. 2022.

\bibitem{liu2023lowoverheadbeam}
W.~Liu, C.~Pan, H.~Ren, F.~Shu, S.~Jin, and J.~Wang, ``Low-overhead beam
  training scheme for extremely large-scale {RIS} in near field,'' \emph{IEEE
  Trans. Commun.}, vol.~71, no.~8, pp. 4924--4940, Aug. 2023.

\bibitem{wang2022jointhybrid3d}
X.~Wang, Z.~Lin, F.~Lin, and L.~Hanzo, ``Joint hybrid {3D} beamforming relying
  on sensor-based training for reconfigurable intelligent surface aided
  {TeraHertz}-based multiuser massive {MIMO} systems,'' \emph{IEEE Sensors J.},
  vol.~22, no.~14, pp. 14\,540--14\,552, Jul. 2022.

\bibitem{huang2023twotimescalebased}
H.~Huang, C.~Zhang, Y.~Zhang, B.~Ning, H.~Gao, S.~Fu, K.~Qiu, and Z.~Han,
  ``Two-timescale-based beam training for {RIS}-aided millimeter-wave
  multi-user {MISO} systems,'' \emph{IEEE Trans. Veh. Technol.}, vol.~72,
  no.~9, pp. 11\,884--11\,897, Sep. 2023.

\bibitem{wang2021jointbeamtraining}
W.~Wang and W.~Zhang, ``Joint beam training and positioning for intelligent
  reflecting surfaces assisted millimeter wave communications,'' \emph{{IEEE}
  Trans. Wireless Commun.}, vol.~20, no.~10, pp. 6282--6297, Oct. 2021.

\bibitem{tian2021fastbeamtracking}
\BIBentryALTinterwordspacing
X.~Tian and Z.~Sun, ``Fast beam tracking for reconfigurable intelligent surface
  assisted mobile {mmWave} networks,'' Feb. 2021. [Online]. Available:
  \url{https://arxiv.org/abs/2102.11414}
\BIBentrySTDinterwordspacing

\bibitem{romanov2021precisesynchronizationmethod}
A.~M. Romanov, F.~Gringoli, and A.~Sikora, ``A precise synchronization method
  for future wireless {TSN} networks,'' \emph{IEEE Trans. Ind. Inf.}, vol.~17,
  no.~5, pp. 3682--3692, May 2021.

\bibitem{rappaport2001wirelesscommunications}
T.~S. Rappaport, \emph{Wireless Communications}.\hskip 1em plus 0.5em minus
  0.4em\relax Prentice Hall PTR, 2001.

\bibitem{laue2021irsassistedactive}
F.~Laue, V.~Jamali, and R.~Schober, ``{IRS}-assisted active device detection,''
  in \emph{Proc. {IEEE} 22nd Int. Workshop Signal Process. Advances in Wireless
  Commun.}, Sep. 2021.

\bibitem{han2021halfpowerbeamwidth}
H.~Han, Y.~Liu, and L.~Zhang, ``On half-power beamwidth of intelligent
  reflecting surface,'' \emph{IEEE Commun. Lett.}, vol.~25, no.~4, pp.
  1333--1337, Apr. 2021.

\bibitem{wei2022codebookdesignbeam}
X.~Wei, L.~Dai, Y.~Zhao, G.~Yu, and X.~Duan, ``Codebook design and beam
  training for extremely large-scale {RIS}: Far-field or near-field?''
  \emph{China Commun.}, vol.~19, no.~6, pp. 193--204, Jun. 2022.

\bibitem{yang2018hierarchicalcodebookbeam}
L.~Yang and W.~Zhang, ``Hierarchical codebook and beam alignment for {UAV}
  communications,'' in \emph{Proc. {IEEE} Globecom Workshops}, Dec. 2018.

\bibitem{camp2002surveymobilitymodels}
T.~Camp, J.~Boleng, and V.~Davies, ``A survey of mobility models for ad hoc
  network research,'' \emph{Wireless Commun. Mob. Comput.}, vol.~2, no.~5, pp.
  483--502, Sep. 2002.

\bibitem{hassan2021keytechnologiesultra}
B.~Hassan, S.~Baig, and M.~Asif, ``Key technologies for ultra-reliable and
  low-latency communication in {6G},'' \emph{{IEEE} Commun. Standards Mag.},
  vol.~5, no.~2, pp. 106--113, Jun. 2021.

\bibitem{wang2020compressedchannelestimation}
P.~Wang, J.~Fang, H.~Duan, and H.~Li, ``Compressed channel estimation for
  intelligent reflecting surface-assisted millimeter wave systems,'' \emph{IEEE
  Signal Process Lett.}, vol.~27, pp. 905--909, May 2020.

\end{thebibliography}

\end{document}